\documentclass[12pt]{article}
\usepackage[utf8]{inputenc}

\newcommand{\ignore}[1]{}
\usepackage{physics}
\usepackage{mathtools}
\usepackage{amsmath}
\usepackage{color}
\usepackage{amsfonts}
\usepackage{amssymb}
\usepackage[mathscr]{euscript}
\usepackage{graphicx}
\usepackage{geometry}
\usepackage{amssymb,epsfig,subfigure}
\usepackage{hyperref}
\usepackage{comment}
\usepackage{tabularx}
\usepackage{bm}
\usepackage{euscript}
\usepackage{graphicx}
\usepackage[dvipsnames]{xcolor}
\usepackage{amsfonts}
\usepackage{exscale}
\usepackage{scalerel}
\usepackage{tensor}
\usepackage{amsbsy}
\usepackage{subfigure}
\usepackage{textcomp}
\usepackage{comment}
\usepackage{hyperref}
\usepackage{slashed}
\usepackage{authblk}
\usepackage{tabularx}
\usepackage{euscript}
\usepackage{graphicx}
\usepackage{color}
\usepackage{amsfonts}
\usepackage{exscale}
\usepackage{amsbsy}
\usepackage{subfigure}
\usepackage{textcomp}
\usepackage{comment}
\usepackage{hyperref}
\usepackage{bm}
\usepackage{wrapfig}
\usepackage[font=footnotesize,labelfont=bf]{caption}
\usepackage{ulem} 
\usepackage{makecell}

\def\ba#1\ea{\begin{align}#1\end{align}}
\def\bg#1\eg{\begin{gather}#1\end{gather}}
\def\bm#1\em{\begin{multline}#1\end{multline}}
\def\bmd#1\emd{\begin{multlined}#1\end{multlined}}

\newcommand{\be}{\begin{equation}}
\newcommand{\ee}{\end{equation}}
\newcommand{\bea}{\begin{eqnarray}}
\newcommand{\eea}{\end{eqnarray}}

\newcommand{\matleft}{\left(\begin{array}}
\newcommand{\matright}{\end{array}\right)}

\newcommand{\ZS}[1]{\textcolor{magenta}{\it [ZS: #1]}}

\newcommand{\eqnref}[1]{Eq.~(\ref{#1})}

\usepackage[numbers,sort&compress]{natbib}
\def\simge{
    \mathrel{\rlap{\raise 0.511ex 
        \hbox{$>$}}{\lower 0.511ex \hbox{$\sim$}}}}

\def\simle{
    \mathrel{\rlap{\raise 0.511ex 
        \hbox{$<$}}{\lower 0.511ex \hbox{$\sim$}}}}

\makeatletter
\renewcommand\section{\@startsection {section}{1}{\z@}%
                                 {-3.5ex \@plus -1ex \@minus -.2ex}
                                   {2.3ex \@plus.2ex}%
                                   {\normalfont\large\bfseries}}
\renewcommand\subsection{\@startsection{subsection}{2}{\z@}%
                                   {-3.25ex\@plus -1ex \@minus -.2ex}%
                                     {1.5ex \@plus .2ex}%
                                     {\normalfont\bfseries}}
\renewcommand\subsubsection{\@startsection{subsubsection}{3}{\z@}%
                                   {-3.25ex\@plus -1ex \@minus -.2ex}%
                                     {1.5ex \@plus .2ex}%
                                     {\normalfont\itshape}}
\makeatother

\def\pplogo{\vbox{\kern-\headheight\kern -29pt
\halign{##&##\hfil\cr&{\ppnumber}\cr\rule{0pt}{2.5ex}&\ppdate\cr}}}
\makeatletter
\def\ps@firstpage{\ps@empty \def\@oddhead{\hss\pplogo}%
  \let\@evenhead\@oddhead 
}
\def\maketitle{\par
 \begingroup
 \def\thefootnote{\fnsymbol{footnote}}
 \def\@makefnmark{\hbox{$^{\@thefnmark}$\hss}}
 \if@twocolumn
 \twocolumn[\@maketitle]
 \else \newpage
 \global\@topnum\z@ \@maketitle \fi\thispagestyle{firstpage}\@thanks
 \endgroup
 \setcounter{footnote}{0}
 \let\maketitle\relax
 \let\@maketitle\relax
 \gdef\@thanks{}\gdef\@author{}\gdef\@title{}\let\thanks\relax}
\makeatother

\hypersetup{
    unicode=false,          
    pdftoolbar=true,        
    pdfmenubar=true,        
    pdffitwindow=false,     
    pdfstartview={FitH},    
    pdftitle={Critical drag},    
    pdfauthor={},     
    pdfsubject={Subject},   
    pdfcreator={},   
    pdfproducer={}, 
    pdfkeywords={keyword1} {key2} {key3}, 
    pdfnewwindow=true,      
    colorlinks=true,       
    linkcolor=OliveGreen, 
    citecolor=NavyBlue,        
    filecolor=magneta,      
    urlcolor=cyan           
}

\numberwithin{equation}{section}

\newcommand*\samethanks[1][\value{footnote}]{\footnotemark}

\usepackage{amsmath}
\usepackage{color}
\usepackage{amsfonts}
\usepackage{amssymb}
\usepackage[mathscr]{euscript}
\usepackage{graphicx}
\usepackage{geometry}
\usepackage{amssymb,epsfig,subfigure}
\usepackage{hyperref}
\usepackage{comment}
\usepackage{tabularx}
\usepackage{bm}
\usepackage{euscript}
\usepackage{graphicx}
\usepackage[dvipsnames]{xcolor}
\usepackage{amsfonts}
\usepackage{exscale}
\usepackage{amsbsy}
\usepackage{subfigure}
\usepackage{textcomp}
\usepackage{comment}
\usepackage{hyperref}
\usepackage{slashed}
\usepackage{authblk}
\usepackage{tabularx}
\usepackage{euscript}
\usepackage{graphicx}
\usepackage{color}
\usepackage{amsfonts}
\usepackage{exscale}
\usepackage{amsbsy}
\usepackage{subfigure}
\usepackage{textcomp}
\usepackage{comment}
\usepackage{hyperref}
\usepackage{bm}
\usepackage{wrapfig}
\usepackage[font=footnotesize,labelfont=bf]{caption}
\usepackage{ulem} 
\usepackage{makecell}
\def\ba#1\ea{\begin{align}#1\end{align}}
\def\bg#1\eg{\begin{gather}#1\end{gather}}
\def\bm#1\em{\begin{multline}#1\end{multline}}
\def\bmd#1\emd{\begin{multlined}#1\end{multlined}}

\usepackage[numbers,sort&compress]{natbib}
\def\simge{
    \mathrel{\rlap{\raise 0.511ex 
        \hbox{$>$}}{\lower 0.511ex \hbox{$\sim$}}}}

\def\simle{
    \mathrel{\rlap{\raise 0.511ex 
        \hbox{$<$}}{\lower 0.511ex \hbox{$\sim$}}}}

\makeatletter
\renewcommand\section{\@startsection {section}{1}{\z@}%
                                 {-3.5ex \@plus -1ex \@minus -.2ex}
                                   {2.3ex \@plus.2ex}%
                                   {\normalfont\large\bfseries}}
\renewcommand\subsection{\@startsection{subsection}{2}{\z@}%
                                   {-3.25ex\@plus -1ex \@minus -.2ex}%
                                     {1.5ex \@plus .2ex}%
                                     {\normalfont\bfseries}}
\renewcommand\subsubsection{\@startsection{subsubsection}{3}{\z@}%
                                   {-3.25ex\@plus -1ex \@minus -.2ex}%
                                     {1.5ex \@plus .2ex}%
                                     {\normalfont\itshape}}
\makeatother

\def\pplogo{\vbox{\kern-\headheight\kern -29pt
\halign{##&##\hfil\cr&{\ppnumber}\cr\rule{0pt}{2.5ex}&\ppdate\cr}}}
\makeatletter
\def\ps@firstpage{\ps@empty \def\@oddhead{\hss\pplogo}%
  \let\@evenhead\@oddhead 
}
\def\maketitle{\par
 \begingroup
 \def\thefootnote{\fnsymbol{footnote}}
 \def\@makefnmark{\hbox{$^{\@thefnmark}$\hss}}
 \if@twocolumn
 \twocolumn[\@maketitle]
 \else \newpage
 \global\@topnum\z@ \@maketitle \fi\thispagestyle{firstpage}\@thanks
 \endgroup
 \setcounter{footnote}{0}
 \let\maketitle\relax
 \let\@maketitle\relax
 \gdef\@thanks{}\gdef\@author{}\gdef\@title{}\let\thanks\relax}
\makeatother

\hypersetup{
    unicode=false,          
    pdftoolbar=true,        
    pdfmenubar=true,        
    pdffitwindow=false,     
    pdfstartview={FitH},    
    pdftitle={},    
    pdfauthor={},     
    pdfsubject={Subject},   
    pdfcreator={},   
    pdfproducer={}, 
    pdfkeywords={keyword1} {key2} {key3}, 
    pdfnewwindow=true,      
    colorlinks=true,       
    linkcolor=OliveGreen, 
    citecolor=NavyBlue,        
    filecolor=magneta,      
    urlcolor=cyan           
}

\numberwithin{equation}{section}

 \newcommand\beal{\begin{equation}\begin{aligned}}
\newcommand\eeal{\end{aligned}\end{equation}}

\usepackage{amsmath, amsthm, amssymb}
\newtheorem{theorem}{Theorem}[section]
\newtheorem{definition}{Definition}[section]

\newtheorem{lemma}[definition]{Lemma}
\newtheorem{corollary}[definition]{Corollary}

\begin{document}

\normalem

\setcounter{page}0
\def\ppnumber{\vbox{\baselineskip14pt
}}

\def\ppdate{
} \date{\today}

\title{Bounds on Renyi entropy growth \\in many-body quantum systems}
\author{Zhengyan Darius Shi}
\affil{\it \small Department of Physics, Massachusetts Institute of Technology, Cambridge, MA 02139, USA}

\maketitle
\begin{abstract}
    We prove rigorous bounds on the growth of $\alpha$-Renyi entropies $S_{\alpha}(t)$ (the Von Neumann entropy being the special case $\alpha = 1$) associated with any subsystem $A$ of a general lattice quantum many-body system with finite onsite Hilbert space dimension. For completely non-local Hamiltonians, we show that the instantaneous growth rates $|S'_{\alpha}(t)|$ (with $\alpha \neq 1$) can be exponentially larger than $|S'_1(t)|$ as a function of the subsystem size $|A|$. For $D$-dimensional systems with geometric locality, we prove bounds on $|S'_{\alpha}(t)|$ that depend on the decay rate of interactions with distance. When $\alpha = 1$, the bound is $|A|$-independent for all power-law decaying interactions $V(r) \sim r^{-w}$ with $w > 2D+1$. But for $\alpha > 1$, the bound is $|A|$-independent only when the interactions are finite-range or decay faster than $V(r) \sim e^{- c\, r^D}$ for some $c$ depending on the local Hilbert space dimension. Using similar arguments, we also prove bounds on $k$-local systems with or without geometric locality. A central theme of this work is that the value of $\alpha$ strongly influences the interplay between locality and entanglement growth. In other words, the Von Neumann entropy and the $\alpha$-Renyi entropies cannot be regarded as proxies for each other in studies of entanglement dynamics. We compare these bounds with analytic and numerical results on Hamiltonians with varying degrees of locality and find concrete examples that almost saturate the bound for non-local dynamics. 
\end{abstract}

\pagebreak
{
\hypersetup{linkcolor=black}
\tableofcontents
}
\pagebreak

\maketitle

\section{Introduction}\label{sec:intro}

Spatial locality is a basic property of most physical models for many-body quantum dynamics. Attempts to quantify the spread of quantum information under spatially local interactions have motivated many breakthroughs in the field. An influential example is the Lieb-Robinson bound, a rigorous speed limit $v_{LR}$ for the ballistic spread of local disturbances, as measured by the increasing norm of commutators between spatially distant operators under Heisenberg evolution\cite{Lieb_Robinson_1972, Nachtergaele_Sims_2006, Chen_Lucas_2019, Tran_Chen_Ehrenberg_Guo_Deshpande_Hong_Gong_Gorshkov_Lucas_2020, Else_Machado_Nayak_Yao_2020}. Besides its conceptual elegance, this speed limit also plays an essential role in applications ranging from correlation decay in many-body ground states and stability of topological order to efficient quantum simulation and bounds on diffusive transport\cite{Hastings_Koma_2006, Hastings_2010, Bravyi_Hastings_Michalakis_2010, Han_Hartnoll_2018,Tran_Guo_Su_Garrison_Eldredge_Foss-Feig_Childs_Gorshkov_2019}. 

The universal nature of Lieb-Robinson bounds comes at a price: because the velocity scale $v_{LR}$ only depends on the microscopic couplings, it does not distinguish between different physical processes and different initial states. Hence, in order to understand more refined consequences of locality, it is important to consider proxies of information propagation beyond the commutator norm. The information measure that we consider in this paper is the dynamics of entanglement entropy between complementary subsystems. For a wide class of models, in the absence of many-body localization, the growth of Von Neumann entropy across any cut has been found to follow a linear profile until saturating to the equilibrium value if the system is initialized in an unentangled state. In the pioneering work of Calabrese and Cardy in \cite{Calabrese_Cardy_2005}, this phenomenon was first proven analytically for the integrable one-dimensional transverse field Ising model and for 1+1D conformal field theories more generally. This initial work was extended to arbitrary rational CFTs in \cite{Asplund_Bernamonti_Galli_Hartman_2015}, to higher-dimensional free scalar field theories in \cite{Cotler_Hertzberg_Mezei_Mueller_2016}, and to generic integrable systems in \cite{Alba_Calabrese_2018}. Meanwhile, numerical studies have found a similar linear growth for non-integrable systems including the XXZ Heisenberg spin chain and the mixed field Ising model \cite{Chiara_Montangero_Calabrese_Fazio_2006,Kim_Huse_2013}. Exact analytic calculations for non-integrable systems have been mostly out of reach, with the exception of strongly interacting holographic theories analyzed in \cite{Hubeny_Rangamani_Takayanagi_2007, Abajo-Arrastia_Aparicio_Lopez_2010, Albash_Johnson_2011, Allais_Tonni_2012, Hartman_Maldacena_2013, Liu_Suh_2014a, Liu_Suh_2014b, Leichenauer_Moosa_2015, Casini_Liu_Mezei_2016} and local random unitary circuits studied in \cite{Nahum_Ruhman_Vijay_Haah_2017, Zhou_Nahum_2019, Keyserlingk_Rakovszky_Pollmann_Sondhi_2018}. The linear growth of Von Neumann entropy appears to be a robust feature across this wide range of models. 

One interesting approach to understanding this linear profile is to prove universal upper bounds on the rate of entanglement growth. Such bounds do not constrain the functional form of the growth curve, but they guarantee that starting from an initial weakly entangled state, the time it takes to increase the entanglement entropies for any bipartition of the system by $\delta S$ is at least linear in $\delta S$. If the intuition of ballistic information propagation is correct, then for a subsystem $A$ with size $|A|$, the growth rate should be proportional to the size of the boundary $|\partial A|$ for geometrically local Hamiltonians, and at most $\mathcal{O}(|A|)$ for non-local Hamiltonians. As shown in Refs.~\cite{Bravyi_2007,Acoleyen_Marien_Verstraete_2013,Marien_Audenaert_Acoleyen_Verstraete_2016}, these expectations are indeed correct for the growth of Von Neumann entropies. However, very little is known about analogous bounds for Renyi entropies. A recent work by Vershynina~\cite{Vershynina_2019} proved some bounds on Renyi entropy growth rates in the context of non-local Hamiltonian dynamics. However, these bounds depend sensitively on the spectrum of the reduced density matrix $\rho_A(t)$ and are hence non-universal. As we will discuss in Section~\ref{sec:new_bound}, this state-dependence also makes it difficult to extend these bounds to Hamiltonians with geometric locality/$k$-locality. 

We are therefore confronted with a conundrum. On the one hand, the membrane picture of entanglement growth suggests that $S_{\alpha > 1}(t)$ and $S_1(t)$ should be qualitatively similar under generic non-integrable local Hamiltonian evolution~\cite{Nahum_Ruhman_Vijay_Haah_2017,jonay2018_membrane,Zhou_Nahum_2019,ZhouNahum2020_membrane_chaotic}. On the other hand, existing techniques place much stronger constraints on $|S'_1(t)|$ than on $S'_{\alpha > 1}(t)$. A natural question thus arises: is there any intrinsic difference between the dynamics of these two types of entropy measures in arbitrary many-body systems? In this work, we prove a series of upper bounds on $|S'_{\alpha}(t)|$ in favor of such a distinction. These bounds remove the state-dependence of the bounds in Ref.~\cite{Vershynina_2019} and generalize easily to Hamiltonians with varying degrees of locality, as we now sketch: \newline

\noindent \textbf{Main result} (sketch): Take an arbitrary initial state $\rho$ on a Hilbert space $(\mathbb{C}^{d_0})^{\otimes V}$ where $V$ is the size of the $D$-dimensional lattice and $d_0$ is the onsite Hilbert space dimension. Given any bipartition of the lattice into $A, \bar A$ with $|A|+|\bar A| = V$, let $V_{\rm min} = \min(|A|, |\bar A|)$, and let $S_{\alpha}(t)$ be the Renyi entropy of the reduced density matrix $\Tr_{\bar A} \rho(t)$ under time evolution generated by a Hamiltonian $H$. Then:
\begin{enumerate}
    \item (Lemma~\ref{lemma:alpha>1} + Lemma~\ref{lemma:stronger_comparableAB}) If $H$ is arbitrary, then $|S'_{\alpha > 1}(t)| \lesssim \text{Poly}\left(||H|| d_0^{V_{\rm min}}\right)$ while $|S'_1(t)| \lesssim \mathcal{O}(||H|| V_{\rm min} \log d_0 )$. The exponential separation between $S'_{\alpha > 1}(t)$ and $S'_1(t)$ \textit{can be saturated}.
    \item (Theorem~\ref{thm:Bound}) If $H$ is geometrically local with power-law decaying interactions $V(r) \sim r^{-w}$ satisfying $w > 2D+1$, then $|S'_1(t)| \leq c \bar h|\partial A|\log d_0$ where $c$ is a constant, $\bar h$ measures the norm of local terms in the Hamiltonian and $|\partial A|$ is the size of the boundary of $A$. We \textit{could not} prove a bound on $|S'_{\alpha>1}(t)|$.
    \item (Theorem~\ref{thm:Bound}) If $H$ is strictly local with interaction range $R$, then $|S'_{\alpha > 1}(t)| \leq b_{\alpha}(R) \bar h |\partial A|$ where $b_{1}(R) \propto R^{D+1}$ and $b_{\alpha > 1}(R) \sim R d_0^{\mathcal{O}(R^D)}$. 
    \item (Corollary~\ref{cor:Bound_klocal}) If $H$ is not geometrically local but $k$-local (i.e. each interaction term only contains $k$ sites), then $|S'_{\alpha>1}(t)| \leq b_{\alpha}(k) ||H_{\partial}||$ where $b_{1}(k) \propto k$, $b_{\alpha > 1}(k) \sim d_0^{k-1}$, and $H_{\partial}$ is the part of the Hamiltonian that couples $A$ and $\bar A$. 
    \item (Corollary~\ref{cor:Bound_klocal+geolocal}) If $H$ is both $k$-local and geometrically local with power-law decaying interactions $V(r) \sim r^{-w}$ satisfying $w > Dk$, then for all $\alpha > 1$, $|S'_{\alpha}(t)| \leq b_{\alpha}(k) \bar h |\partial A| \log d_0$ where $b_{\alpha}(k)$ has the same scaling as in part 4. 
\end{enumerate}
The structure of these bounds highlights conceptual distinctions between $\alpha$-Renyi entropies and Von Neumann entropy: despite the fact that they follow similar growth profiles in most examples known to date, the general bounds that we are able to prove for $S'_{\alpha \geq 1}(t)$ are much weaker than the corresponding bounds for $S'_1(t)$ under identical locality assumptions. We emphasize however that with the exception of (1), it is \textit{not known whether these bounds can be saturated}. We hope the question of saturation will be more extensively studied in the future. 

The rest of the paper will be organized as follows. Section~\ref{subsec:bounds_history} provides a summary of known theorems on entanglement growth bounds for general (potentially non-local) Hamiltonians \cite{Bennett_Harrow_Leung_Smolin_2003, Bravyi_2007, Acoleyen_Marien_Verstraete_2013}. We will not discuss the proofs of these theorems but instead introduce a family of operator norm inequalities (first conjectured in \cite{Audenaert_Kittaneh_2012} and partially proven in \cite{Vershynina_2019}) that play an essential role in the proofs \cite{Audenaert_Kittaneh_2012, Lieb_Vershynina_2013, Acoleyen_Marien_Verstraete_2013, Vershynina_2019}. These operator norm inequalities directly give rise to state-dependent bounds on $|S'_{\alpha \geq 1}(t)|$ that appeared in Ref.~\cite{Vershynina_2019}. We continue to work with non-local Hamiltonians in Section~\ref{subsec:improve_vershynina} and show how to improve the bounds in Ref.~\cite{Vershynina_2019} by removing their state-dependence (this is part 1 of the main result). In Section~\ref{sec:randomGUE_example}, these bounds are then shown to be almost saturated by a random GUE Hamiltonian\footnote{The meaning of ``almost saturated" will be made precise later.}. Starting from Section~\ref{sec:new_bound}, we add additional structures to the Hamiltonian and prove bounds on $|S'_{\alpha}(t)|$ for systems with geometric and/or $k$-locality (parts 2-5 of the main result), emphasizing the interplay between locality and Renyi index $\alpha$. This abstract discussion is complemented by a direct comparison of the bounds with a few examples in Section~\ref{sec:new_bound_example}. We conclude in Section~\ref{sec:discussion} with a conceptual summary and some open questions.

\section{Entanglement growth bounds for non-local Hamiltonians}

Consider a many-body system decomposed into four distinct parts $aABb$ and a Hamiltonian $H = I_a \otimes H_{AB} \otimes I_b$ acting on the full Hilbert space. For any initial density matrix $\rho(0)$, we consider the family of $\alpha$-Renyi entropies $S_{\alpha}$ and $q$-Tsallis entropies $T_q$ associated with the subsystem $aA$ (since the state is pure, this is equivalent to the entropy associated with the subsystem $Bb$)
\begin{equation}
    S_{\alpha}(t) = \frac{1}{1-\alpha} \log \Tr_{aA} [\Tr_{Bb} \rho(t)]^{\alpha} \,,\quad T_q(t) = \frac{1}{q-1} \left(1 - \Tr_{aA} [\Tr_{Bb} \rho(t)]^q \right) \,.
\end{equation}
The basic goal is to bound $|S_{\alpha}'(t)|$ and $|T_q'(t)|$. From now on, we refer to $AB$ as the interacting subsystem and $ab$ as ancillas. A concrete realization of this situation would be a quantum processor $aA$ interacting with its environment $Bb$. $AB$ contain the degrees of freedom near the environment-processor interface, and $ab$ contain the degrees of freedom farther away. The spread of entanglement can then be viewed as the spread of quantum noise from the environment into the computational qubits. With this setup in mind, we turn to a brief summary of existing bounds in the literature. We focus on $\alpha$-Renyi entropies, anticipating that analogous results for $q$-Tsallis entropies follow as simple corollaries.

\subsection{A brief history of entanglement growth bounds}\label{subsec:bounds_history}

The earliest works on this subject considered the simplest scenario where $ab = \emptyset$, and $A,B$ are two qubits. For an arbitrary Hamiltonian acting on $AB$, \cite{Dur_Vidal_Cirac_Linden_Popescu_2001,Childs_Leung_Verstraete_Vidal_2003} showed that $|S'_1(t)| \leq \beta ||H||$ where $\beta \approx 1.9123$ is a system-independent constant and $||\cdot||$ is the operator norm. This result was soon generalized to qudit systems by Bravyi, giving a bound $|S'_1(t)| \leq c ||H|| \log d$ where $d = \min\{d_A, d_B\}$ and $d_A,d_B$ are the local Hilbert space dimensions of $A, B$ respectively\cite{Bravyi_2007}. 

In the presence of ancillas, the story is richer. Since $H$ does not couple $AB$ to $ab$, one might expect at first that the bound should not depend on the properties of the ancillas. However, any amount of intrinsic entanglement between $a,A$ and between $B,b$ present in the initial density matrix $\rho(0)$ can act as a catalyst for entanglement growth across the boundary between $aA$ and $Bb$. Despite this subtlety, ancilla-independent bounds have been established and improved upon over the past decade. The arguments in \cite{Bravyi_2007} already implied the loose bound $|S'_1(t)| \leq c ||H|| d^4$. The RHS of the bound was subsequently improved to $c ||H|| d$ in \cite{Lieb_Vershynina_2013} and to $c ||H|| \log d$ in \cite{Acoleyen_Marien_Verstraete_2013}. Though the $\log d$ scaling is known to be optimal, the constant prefactor has been improved from 18 to 2 by \cite{Audenaert_2014, Ning2016}. Up to this point, all the bounds are for Von Neumann entropies in closed quantum systems. One natural question to ask is the existence and scaling of such bounds for more general entropy measures, such as the $\alpha$-Renyi entropies and $q$-Tsallis entropies introduced before. This question was taken up and partially resolved by Vershynina in \cite{Vershynina_2019}. To get a concrete understanding of this work, we first simplify $S'_{\alpha}(t)$ as
\begin{equation}\label{eq:relating_timederivative_to_norm} 
    \begin{aligned}
    S'_{\alpha}(t) &= -\frac{\alpha i}{\alpha-1} \frac{\Tr_{aA} \rho_{aA}(t)^{\alpha-1} \Tr_{Bb} [H, \rho(t)]}{\Tr_{aA} \rho_{aA}(t)^{\alpha}} \\
    &= -\frac{\alpha i}{\alpha-1} \frac{\Tr_{aABb} H [\rho(t), \rho_{aA}(t)^{\alpha-1} \otimes I_{Bb}]}{\Tr_{aA} \rho_{aA}(t)^{\alpha}} \\
    &= -\frac{\alpha i}{\alpha-1} \frac{\Tr_{aAB} H [\rho_{aAB}(t), \rho_{aA}(t)^{\alpha-1} \otimes I_{B}]}{\Tr_{aA} \rho_{aA}(t)^{\alpha}}\,,
    \end{aligned}
\end{equation}
where in the last line, we used the fact that $H = I_a \otimes H_{AB} \otimes I_b$ to trace out the ancilla subsystem $b$. If we rescale $H$ by its operator norm $||H||$, then the numerator of the above expression becomes bounded by the Schatten 1-norm $||[\rho_{aAB}(t), \rho_{aA}(t)^{\alpha-1} \otimes I_{B}]||_1$ of an operator on the full Hilbert space\footnote{The Schatten 1-norm is sometimes also referred to as the nuclear norm/trace norm.}. Therefore, bounding entanglement growth is essentially equivalent to bounding this Schatten 1-norm under suitable assumptions on the density matrices. We now state the most general version of such a bound proven by Vershynina in \cite{Vershynina_2019}:
\begin{lemma}\label{lemma:Vershynina_ineq}
    (Vershynina 2019): Consider a general function $f: [0,1] \rightarrow \mathbb{R}$ satisfying the properties below: 
    \begin{enumerate}
        \item f increases monotonically.
        \item For $0 < y < x \leq 1$ and for some particular $0<p<1$, 
        \begin{equation}
            x^{-1/2} y^{1/2} \left[f(x) - f(y)\right] \leq p^{1/2} \left[f(1) - f(p)\right] \quad \text{if } f(x) - f(y) > f(1) - f(p) \,.
        \end{equation}
    \end{enumerate}
    Then we the following bound holds
    \begin{equation}
        ||[X, f(Y)]||_1 \leq 9 \min \big\{p\left[f(1) - f(p)\right], (1-p)\left[f(1) - f(1-p)\right]\big\}
    \end{equation}
\end{lemma}
In the special case $f(y) = \log y$, Lemma~\ref{lemma:Vershynina_ineq} implies the inequality of \cite{Acoleyen_Marien_Verstraete_2013} which is useful for bounding Von Neumann entropies
\begin{equation}
    ||[X, \log Y]||_1 \leq 9 \min \{p \log (\frac{1}{p}), (1-p) \log (\frac{1}{1-p})\}
\end{equation}
Similarly, plugging in $f(y) = y^{\alpha-1}$ for $1 < \alpha$ gives bounds for $S'_{\alpha}(t)$. For $\alpha \neq 1$, there is no logarithmic factor in $p$ and the bound becomes stronger in the limit $p \rightarrow 0$:
\begin{enumerate}
    \item For $1 < \alpha < 2$, we have $||[X, Y^{\alpha-1}]||_1 \leq 9 (1-p)[1 - (1-p)^{\alpha-1}] = 9 (\alpha - 1) p + \mathcal{O}(p^2)$. 
    \item For $2 \leq \alpha$, we have $||[X, Y^{\alpha-1}]||_1 \leq 9 p (1-p)^{\alpha-1}$. 
\end{enumerate}
For $\alpha < 1$ on the other hand, since $f(y) = y^{\alpha-1}$ does not fulfill the assumptions of Lemma~\ref{lemma:Vershynina_ineq}, there is no obvious bound. 

To constrain entanglement growth, we start with \eqnref{eq:relating_timederivative_to_norm} and trace out the ancilla subsystem $b$. Applying the norm inequalities in Lemma~\ref{lemma:Vershynina_ineq} with $X = \rho_{aAB}/D_{B}^2$, $Y = \rho_{aA} \otimes I_{B}/D_{B}$ and $p = D_{B}^{-2}$ where $I_{B}$ is the identity operator acting on the $B$ subsystem Hilbert space with dimension $D_{B}$, we find that 
\begin{enumerate}\label{Vershynina_entropybounds}
    \item For $1 < \alpha < 2$, $|S_{\alpha}'(t)| \leq \frac{9\alpha}{\alpha-1} ||H|| D_B^{1+\alpha} \left[\Tr_{aA} \rho_{aA}(t)^{\alpha}\right]^{-1} (1 - D_B^{-2}) [1 - (1 - D_B^{-2})^{\alpha -1}]$. 
    \item For $2 \leq \alpha$, $|S_{\alpha}'(t)| \leq \frac{9\alpha}{\alpha-1} ||H|| D_B^{1+\alpha} \left[\Tr_{aA} \rho_{aA}(t)^{\alpha}\right]^{-1} (D_B^{-2} - D_B^{-2\alpha})$. 
\end{enumerate}
We now make a physical observation. Although the bound for Von Neumann entropy is completely state-independent, the bounds for $\alpha$-Renyi entropies depend sensitively on the state via the factor of $\Tr_{aA} \rho_{aA}(t)^{\alpha}$ in the denominator. For $\alpha > 1$, we have the inequality $\Tr_{aA} \rho_{aA}(t)^{\alpha} \leq 1$ and thus $\frac{1}{\Tr \rho_{aA}(t)^{\alpha}} \geq 1$. But this is not helpful because the inequality goes in the wrong direction. In fact, under generic dynamics, $S_{\alpha}(t)$ vanishes at $t = 0$ and saturates to an equilibrium value that is extensive in the size $|aA|$ of the subsystem $aA$. This means that for large $aA$, $\Tr_{aA} \rho_{aA}(t)^{\alpha} =  e^{(1-\alpha) S_{\alpha}(t)}$ can be exponentially small in $|aA|$, implying that the bound is exponentially looser than the bound for $\alpha = 1$ which is independent of $|aA|$.

\subsection{An improvement of Vershynina's bounds}\label{subsec:improve_vershynina}

The preceding discussion leaves open an important question: can we tweak Vershynina's bounds to remove the dependence of the entropy growth rate $|S'_{aA,\alpha}(t)|$ on the subsystem density matrix $\rho_{aA}(t)$? In this section, we answer this question in the affirmative via a simple operator norm inequality formulated as follows:
\begin{lemma}\label{lemma:alpha>1}
    Suppose that $0 \leq X \leq Y \leq I$ and $1< \alpha$, then 
    \begin{equation}
        ||X Y^{\alpha-1}||_1 \leq ||Y^{\alpha}||_1 = \Tr Y^{\alpha} \,.
    \end{equation}
\end{lemma}
\begin{proof}
    We apply the H\"older inequality for Schatten p-norms:
    \begin{equation}
        ||X Y^{\alpha-1}||_1 \leq ||X||_p ||Y^{\alpha-1}||_q \quad \quad \forall \quad \frac{1}{p} + \frac{1}{q} = 1 \,, \quad p, q \geq 1 \,.
    \end{equation}
    Choosing $p = \alpha, q = \frac{\alpha}{\alpha-1} \geq 1$ and using the fact that $X^{\alpha}, Y^{\alpha-1}$ are both positive semidefinite, we arrive at
    \begin{equation}
        ||X Y^{\alpha-1}||_1 \leq (\Tr X^{\alpha})^{\frac{1}{\alpha}} (\Tr Y^{\alpha})^{\frac{\alpha-1}{\alpha}} \,.
    \end{equation}
    Finally, since $\Tr f(X)$ is monotonically increasing whenever $f: \mathbb{R} \rightarrow \mathbb{R}$ is monotonically increasing, we obtain the desired inequality
    \begin{equation}
        ||X Y^{\alpha-1}||_1 \leq (\Tr Y^{\alpha})^{\frac{1}{\alpha} + \frac{\alpha-1}{\alpha}} = \Tr Y^{\alpha} = ||Y^{\alpha}||_1 \,.
    \end{equation}
\end{proof}
We remark that the same inequality \textit{fails to hold} when $\alpha < 1$. The basic problem is that when $\alpha < 1$, $Y^{\alpha-1}$ can have very large eigenvalues $y_i^{\alpha-1}$ whenever $Y$ has eigenvalues $y_i$ close to $0$. When $X,Y$ commute, every large eigenvalue $y_i^{\alpha-1}$ is balanced by a small eigenvalue $x_i \leq y_i$ and the inequality continues to hold
\begin{equation}
    ||XY^{\alpha-1}||_1 = \sum_i x_i y_i^{\alpha-1} \leq \sum_i y_i^{\alpha} = ||Y^{\alpha}||_1 \,.
\end{equation}
But in the general case where $X$ and $Y$ do not commute, the above argument fails. Even though $X \leq Y$ as operators, the eigenvectors $\ket{y_i}$ of $Y$ can overlap with many eigenvectors $\ket{x_j}$ of $X$ with $x_j > y_i$. As a result, the LHS of the bound could receive large contributions from $\ket{y_i}$ while the RHS only receives small contributions since $\alpha > 0$. In the physical cases of interest to us, $Y$ is proportional to the reduced density matrix of a subsystem, whose spectrum necessarily contains many eigenvalues close to $0$. Therefore, we expect the bound to be violated in general. This violation has been verified numerically in random samples of operators satisfying $0 < X \leq Y \leq I$.  

With the above discussion in mind, we focus on the $\alpha > 1$ case. Using Lemma~\ref{lemma:alpha>1} in \eqnref{eq:relating_timederivative_to_norm}, we see that
\begin{equation}\label{eq:bound_nonlocal_independentA}
    \begin{aligned}
    \left|S'_{\alpha}(t)\right| &\leq \frac{\alpha ||H||}{|\alpha-1|} \frac{||\left[\rho_{aAB}(t), \rho_{aA}(t)^{\alpha-1} \otimes I_B\right]||_1}{\Tr_{aA} \rho_{aA}(t)^{\alpha}} \\
    &\leq \frac{2\alpha ||H||}{|\alpha-1|} \frac{D_B^{1-\alpha} D_B^{1+\alpha}}{\Tr Y^{\alpha}} ||X Y^{\alpha-1}||_1 \leq \frac{2\alpha}{|\alpha-1|} \frac{D_B^2}{\Tr Y^{\alpha}} \Tr Y^{\alpha} = \frac{2\alpha}{|\alpha-1|} ||H|| D_B^2 \,,
    \end{aligned}
\end{equation}
where we again defined $X = \rho_{aAB}/D_{B}^2$ and $Y = \rho_{aA} \otimes I_{B}/D_{B}$. Note that the cancellation of $\Tr Y^{\alpha}$ in the numerator and denominator leads to a final bound which is independent of the reduced density matrix $\rho_{aA}(t)$. This bound is ideal when the size of $aA$ is much larger than $B$. However, when $aA$ is comparable to $B$ in size, the bound can be improved further. We now state and prove the optimal inequality that we could find:
\begin{lemma}\label{lemma:stronger_comparableAB}
    If $\alpha > 1$, $0 \leq X \leq Y \leq I$, and $||X||_1 = p$, then 
    \begin{equation}
        ||X Y^{\alpha-1}||_1 \leq p (\Tr Y^{\alpha})^{\frac{\alpha-1}{\alpha}} \,.
    \end{equation}
\end{lemma}
\begin{proof}
    Using H\"older's inequality with $p = 1, q = \infty$, we have
    \begin{equation}
        ||XY^{\alpha-1}||_1 \leq ||X||_1 ||Y^{\alpha-1}||_{\infty} = p ||Y^{\alpha-1}||_{\infty} \,.
    \end{equation}
    Since Schatten p-norms are decreasing in $p$ (i.e. $p_1 \geq p_2$ implies $||A||_{p_1} \leq ||A||_{p_2}$), we conclude that
    \begin{equation}
        ||XY^{\alpha-1}||_1 \leq p ||Y^{\alpha-1}||_{\frac{\alpha}{\alpha-1}} = p (\Tr Y^{\alpha})^{\frac{\alpha-1}{\alpha}} \,,
    \end{equation}
    which is the desired inequality.
\end{proof}
The inequality in Lemma~\ref{lemma:stronger_comparableAB} implies a bound on the entropy
\begin{equation}
    \left|S_{\alpha}'(t)\right| \leq \frac{2\alpha}{|\alpha-1|} \frac{D_B^2}{\Tr Y^{\alpha}} D_B^{-2} (\Tr Y^{\alpha})^{\frac{\alpha-1}{\alpha}} = \frac{2\alpha}{|\alpha-1|} (\Tr Y^{\alpha})^{-1/\alpha} \,.
\end{equation}
If we rewrite $\Tr Y^{\alpha}$ in terms of the subsystem density matrix $\rho_{aA}(t)$ as
\begin{equation}
    \Tr Y^{\alpha} = D_B^{1-\alpha} \Tr_{aA} \rho_{aA}(t)^{\alpha} \geq D_B^{1-\alpha} D_{aA}^{1-\alpha} \,,
\end{equation}
we find a new bound 
\begin{equation}\label{eq:bound_nonlocal_dependentA}
    \left|S_{\alpha}'(t)\right| \leq \frac{2\alpha ||H||}{|\alpha-1|} \left(D_B^{1-\alpha} D_{aA}^{1-\alpha}\right)^{-1/\alpha} \leq \frac{2\alpha}{|\alpha-1|} ||H|| (D_{aA} D_B)^{\frac{\alpha-1}{\alpha}} \,.
\end{equation}
When $D_{aA}, D_B$ are comparable, the above bound scales as $D_B^{\frac{2(\alpha-1)}{\alpha}}$, which is much smaller than the RHS of \eqnref{eq:bound_nonlocal_independentA} for all $\alpha > 1$. We will compare this latter bound against explicit examples in Section~\ref{sec:randomGUE_example}. 

In summary, by extending the results of Ref.~\cite{Vershynina_2019}, we can place useful bounds on the time derivatives of $\alpha$-Renyi entropies $S_{\alpha}(t)$ for all $\alpha > 1$. Surprisingly, our bounds for $\alpha > 1$ scale very differently from the existing bounds for $\alpha = 1$. In fact, for a subsystem $A$ of size $|A|$, the bound for $\left|S'_1(t)\right|$ is linear in $|A|$ while the bound for $\left|S'_{\alpha}(t)\right|$ is exponentially large in $|A|$ for all $\alpha \neq 1$. A saturation of these bounds would imply a fundamental difference between the dynamics of Renyi entropies and Von Neumann entropies. This will be the focus of Section~\ref{sec:randomGUE_example}.

\subsection{Random GUE dynamics: an example of exponential separation between Renyi entropy and Von Neumann entropy derivatives}\label{sec:randomGUE_example}

The preceding discussion raises a natural question: is the exponential separation between the maximal growth rate of different $\alpha$-Renyi entropies merely an artifact of the existing proof techniques? We give a negative answer to this question by analyzing a concrete model of random Hamiltonian evolution that shows exponentially large Renyi entropy growth rates~\cite{You_Gu_2018}. 

Consider a bipartite qudit system $AB$ with total size $V$ such that $V/2 = |A| = |B|$. Define a GUE ensemble of Hamiltonians acting on the Hilbert space $AB$ with total dimension $D = d^V$:
\begin{equation}
    P(H) \sim e^{-\frac{D}{2} \Tr H^2} \,,\quad H^{\dagger} = H \,.
\end{equation}
For each Hamiltonian $H$, define the time-evolved density matrix $\rho^{(H)}(t)$ and the reduced density matrix $\rho^{(H)}_A(t) = \Tr_B \rho^{(H)}(t)$. The quantities to study are the ensemble-averaged Von Neumann entropy and Renyi entropies
\begin{equation}
    \begin{aligned}
    \bar S_1(t) &= \int d H P(H) [-\Tr \rho^{(H)}_A(t) \log \rho^{(H)}_A(t)]\,, \\\bar S_{\alpha>1}(t) &= \int dH P(H) [\frac{1}{1-\alpha} \log \Tr \rho^{(H)}_A(t)^{\alpha}] \,.
    \end{aligned}
\end{equation}
To apply the entropy growth bounds in the previous section to this model, we need to estimate the norm of a GUE random Hamiltonian. It is well known that the GUE eigenvalue distribution converges almost surely to a Wigner semicircle supported on $[-2,2]$. This provides a lower bound on the operator norm $||H|| \geq 2 - o(1)$. The complementary upper bound is provided by a general theorem:
\begin{theorem}
    (Bai-Yin 1988 \cite{Bai_Yin_1988}): Let $M_{ij}$ be a rank-$D$ random matrix with independent mean-zero entries obeying the second-moment bounds
    \begin{equation}
        \sup_i \sum_{j=1}^D \ev{|M_{ij}|^2}, \sup_j \sum_{i=1}^D \ev{|M_{ij}|^2} \leq K^2 
    \end{equation}
    and the fourth moment bounds
    \begin{equation}
        \sum_{i,j=1}^D \ev{|M_{ij}|^4} \leq \infty
    \end{equation}
    for some $K > 0$. Then $\ev{||M||} \leq K + o(1)$. 
\end{theorem}
For the case of GUE random matrices, the fourth moment bound is indeed satisfied and the second moment bounds hold with $K = 2$. Combining the upper bound with the lower bound from semicircle law gives $||H|| = 2 + o(1)$. By measure concentration results in random matrix theory, sample to sample variations from $||H|| = 2$ decay exponentially with the rank of the matrix and hence doubly exponentially with $V$~\cite{alon2002concentration,meckes2004concentration}. This means that the bounds on ensemble-averaged entropies apply also to the entropy of a single realization with errors that vanish as $V \rightarrow \infty$. Hence, we can plug $||H|| = 2$ into the entropy bounds in Section~\ref{subsec:improve_vershynina}:
\begin{equation}
    |S_{\alpha}'(t)| \lesssim \begin{cases} \frac{2\alpha}{|\alpha-1|} ||H|| D_B^2 & \alpha < 1 \\ c ||H|| \log D_B & \alpha = 1 \\ \frac{2\alpha}{|\alpha-1|} ||H|| D_B^{\frac{2(\alpha-1)}{\alpha}} & \alpha > 1 \end{cases} \,.
\end{equation}
In this random Hamiltonian model, we are not able to compute the Von Neumann entropy analytically. However, the ensemble-averaged Renyi entropies are accessible in the large $D$ limit. To illustrate our point, it is sufficient to consider the ensemble-averaged second Renyi entropy, which was computed analytically in Ref.~\cite{You_Gu_2018} to be
\begin{equation}
    \bar S_2(t) = - \log \left[R(t) + [1 - R(t)]\left(\frac{1}{D_A} + \frac{1}{D_B}\right)\right] = - \log \left(R(t) + [1 - R(t)] \cdot \frac{2}{d^{L/2}}\right) \,,
\end{equation}
where $R(t) = J_1(2t)^4/t^4$ and $J_a$ is the order-$a$ Bessel function of the first kind. As a sanity check, $\bar S_2(0) = - \log (1 + 0) = 0$ and $\bar S_2(\infty) = - \log (2d^{-V/2}) = \frac{V}{2} \log d - \log 2$, consistent with the expected equilibrium value for the second Renyi entropy. However, unlike in local random unitary circuits, $\bar S_2(t)$ reaches the saturation value whenever $R(t) \propto J_1(2t) = 0$. The earliest time at which this happens is $t_* = \mathcal{O}(1)$ independent of the system size. Therefore, in the large $V$ limit, $\bar S_2(t)$ has to reach an $\mathcal{O}(V)$ value in $\mathcal{O}(1)$ time, implying that $\max_t |\frac{d}{dt} \bar S_2(t)| \rightarrow \infty$ as $V \rightarrow \infty$. To understand the scaling of this diverging derivative with $V$, we explicitly compute the derivative
\begin{equation}
    \frac{d \bar S_2(t)}{dt} = \frac{8 (d^{V/2} - 4) J_1(2t)^3 J_2(2t)}{4 t^4 + (d^{V/2} - 4) J_1(2t)^4} \,.
\end{equation}
In the limit of large $V$, the maximum of the derivative is achieved near $t = t_*$ where $J_1(2t_*) = 0$ and $J_2(2t_*) \neq 0$. This allows us to approximate $J_1(2t) \approx J_1'(2t) 2 (t-t_*)$ and $J_2(2t) \approx J_2(2t_*)$ and simplify the second derivative to
\begin{equation}
    \begin{aligned}
    \frac{d^2 \bar S_2(t)}{dt^2} &\approx 8 (d^{V/2} - 4) J_1'(2t_*)^3 J_2(2t_*) \cdot 2^3 \frac{d}{dt} \frac{(t-t_*)^3}{4t_*^4 + (d^{V/2} - 4) J_1'(2t_*) 2^4 (t-t*)^4} \\
    &\propto \frac{3 (t-t_*)^2 [4t_*^4 + (d^{V/2} - 4) J_1'(2t_*) 2^4 (t-t*)^4] -4 (d^{V/2} - 4) J_1'(2t_*) 2^4 (t-t*)^6}{[4t_*^4 + (d^{V/2} - 4) J_1'(2t_*) 2^4 (t-t*)^4]^2} \,.
    \end{aligned}
\end{equation}
Setting the second derivative to zero gives an asymptotic identity
\begin{equation}
    12 (t-t_*)^2 t_*^4 = (d^{V/2} - 4) J_1'(2t_*) 2^4 (t-t_*)^6 \quad \rightarrow \quad (t-t_*)^{-1} \sim (d^{V/2})^{1/4} = D^{1/8} \sim e^{\frac{\log d}{8} V} \,,
\end{equation}
which implies that
\begin{equation}
    \max_t \left|\bar S'_2(t)\right| \approx \left|\frac{8 J_1'(2t_*)^3 2^3 (t-t_*)^3 J_2(2t_*)}{J_1'(2t_*)^4 2^4 (t-t_*)^4}\right| \approx \left|\frac{4 J_2(2t_*)}{J_1'(2t_*)} (t-t_*)^{-1}\right| \sim e^{\frac{\log d}{8} V} \,.
\end{equation}
The $V$-scaling of the above expression has been verified numerically. This computation demonstrates that $\max |\bar S'_2(t)| \geq D^{1/8} = e^{\frac{\log d}{8} V} \gg V$ as $V \rightarrow \infty$. Hence the bound on $|S'_2(t)|$ and the value of $\max |\bar S'_2(t)|$ realized in this random GUE Hamiltonian model are \textit{both exponentially larger than the bound on $|S'_1(t)|$}. We say that $\max |\bar S'_2(t)| \sim D^{1/8}$ almost saturates the bound $D^{1/2}$, in the sense that they both scale as $e^{c V}$, but with different values of $c$. It is an open question whether $D^{1/2}$ is asymptotically tight at large $D$. 

\section{Hamiltonians with more structure: bounds for systems with geometric locality and $k$-locality}\label{sec:new_bound}

\subsection{Bounds for systems with geometric locality}\label{subsec:bound_geometric_local}

The entropy bounds in previous sections are applicable to systems that can be partitioned into $aABb$ so that the Hamiltonian $H$ only acts on the subsystem $AB$. This is very different from typical many-body quantum systems where $H$ acts on the entire Hilbert space $aABb$. If we apply the existing bounds directly to such systems, the Von Neumann and Renyi entropy growth rates for the subsystem $aA$ scale as $||H|| \log D_{Bb}, ||H|| \text{Poly}(D_{Bb})$ respectively. Both bounds diverge as $D_{Bb} = d^{V/2} \rightarrow \infty$, making them too loose to be useful. In this section, we show using elementary arguments that the bounds can in fact be strengthened to give a finite rate for both $S_1(t)$ and $S_{\alpha \neq 1}(t)$ assuming the Hamiltonian is local in an appropriate sense. In the case of Von Neumann entropies, the optimal bound is already obtained in Proposition 4 of Ref.~\cite{Marien_Audenaert_Acoleyen_Verstraete_2016} and we only review a sketch of the proof. Our focus will be on several new bounds which apply to $|S'_{\alpha}(t)|$ and highlight the distinct locality assumptions needed for different values of $\alpha$. 
In what follows, we formulate the part of our argument that applies to all $\alpha$ as Lemma~\ref{lemma:key} and then introduce the main result Theorem~\ref{thm:Bound} that differentiates $\alpha \neq 1$ from $\alpha = 1$. Our strategy is a direct adaptation of the techniques used in the proof of Proposition 4 in Ref.~\cite{Marien_Audenaert_Acoleyen_Verstraete_2016}, augmented (in the case of $\alpha \neq 1$) by the new operator inequalities in Lemma~\ref{lemma:alpha>1}. 

Consider a $D$-dimensional lattice $\Xi$ equipped with a metric $d$ and let $B_i(R)$ denote a ball of radius $R$ around site $i \in \Xi$. The most general Hamiltonian on the lattice can be written as
\begin{equation}
    H = \sum_{i \in \Xi, r \in \mathbb{N}} h_i(r) \,,
\end{equation}
where $h_i(r)$ is a sum of interaction terms that satisfy
\begin{equation}
    \mathrm{supp}\left[h_i(r)\right] \not\subset B_i(r-1)\,, \quad \mathrm{supp} \left[h_i(r)\right] \subset B_i(r) \,.
\end{equation}
To study entropy growth, we choose a bipartition of the lattice $\Xi$ into $A, \bar A$, and define $H_{\partial}$ to be the restricted sum over all $h_{i}(r)$ with support intersecting both $A$ and $\bar A$. With this basic setup in mind, we can simplify the entropy growth rate using elementary manipulations that apply for all $\alpha$:
\begin{lemma}\label{lemma:key}
    For any choice of bipartition, the Renyi entropy growth rates can be written in the following form
    \begin{equation}
        \left|S'_{\alpha}(t)\right| = \frac{\alpha}{|\alpha-1|} \left|\frac{\Tr \left(H_{\partial} [\rho(t), \rho_{A}(t)^{\alpha-1} \otimes I_{\bar A}]\right)}{\Tr_A \rho_{A}(t)^{\alpha}} \right| \,. 
    \end{equation}
    Taking the $\alpha \rightarrow 1$ limit, we have a corresponding formula for the Von Neumann entropy:
    \begin{equation}
        \left|S'_1(t)\right| = \left|\Tr \left(H_{\partial}\left[\rho(t), \log \rho_{A}(t) \otimes I_{\bar A}\right]\right)\right| \,.
    \end{equation}
\end{lemma}
\begin{proof}
    From \eqnref{eq:relating_timederivative_to_norm}, we know that
    \begin{equation}
        \begin{aligned}
        S'_{\alpha}(t) &= -\frac{\alpha i}{\alpha -1} \frac{\Tr_{A\bar A} H \left[\rho(t), \rho_A(t)^{\alpha-1} \otimes I_{\bar A}\right]}{\Tr_A \rho_A(t)^{\alpha}} \,.
        \end{aligned}
    \end{equation}
    Fixing $\rho(t)$, we see that the RHS is linear in $H$ and we can consider the effect of each term $h_{i}(r)$ separately. When $h_{i}(r)$ is completely supported in $A$, we can factorize $h_{i}(r) = h_{i,A}(r) \otimes I_{\bar A}$ so that 
    \begin{equation}
        \Tr_{\bar A} h_i(r) \left[\rho(t), \rho_A(t)^{\alpha-1} \otimes I_{\bar A}\right] = h_{i,A}(r) \left[\rho_A(t), \rho_A(t)^{\alpha-1}\right] = 0 \,.
    \end{equation}
    On the other hand, when $h_{i}(r)$ is completely supported in $\bar A$, by the cyclicity of the trace, we have
    \begin{equation}
        \Tr_{\bar A} h_i(r) \left[\rho(t), \rho_A(t)^{\alpha-1} \otimes I_{\bar A}\right] = \Tr_{\bar A} \rho(t)[I_A \otimes h_{i,\bar A}(r), \rho_A(t)^{\alpha-1} \otimes I_{\bar A}] = 0 \,.
    \end{equation}
    Hence, terms in $H$ that are completely supported in $A$ or $\bar A$ do not contribute to the entropy derivative and we can replace $H$ with $H_{\partial}$ inside the trace. 
    This directly implies
    \begin{equation}
        \begin{aligned}
        \left|S'_{\alpha}(t)\right| &= \frac{\alpha}{|\alpha-1|} \left|\frac{\Tr \left(H_{\partial} [\rho(t), \rho_{A}(t)^{\alpha-1} \otimes I_{\bar A}]\right)}{\Tr_A \rho_{A}(t)^{\alpha}} \right|  \,.
        \end{aligned}
    \end{equation}
    The $\alpha \rightarrow 1$ limit follows immediately. 
\end{proof}
Before stating the bounds, we need to define two notions of locality for qudit Hamiltonians, given a metric $d$ on the lattice $\Xi$: 
\begin{enumerate}
    \item \textit{Power-law local} with decay exponent $w$ if $||h_{i}(r)|| \lesssim \frac{\bar h}{r^{w}}$ for some constant $\bar h$.
    \item \textit{Stretched-exponentially local} with decay exponent $w > 0$ and decay length $\xi$ if there is a finite constant $\bar h$ such that $||h_{i}(r)|| \lesssim \bar h e^{- (\frac{r}{\xi})^{w}}$. We will often refer to the cases $0 < w < 1, w = 1, w > 1$ as sub-exponential, exponential, and super-exponential respectively. 
\end{enumerate}
With this definition in mind, and combining Lemma~\ref{lemma:key} with existing results quoted in the previous section, we have the following theorem:
\begin{theorem}\label{thm:Bound}
    Consider the same setup as in Lemma~\ref{lemma:key} and in the ensuing discussion of locality. Let the spatial dimension of the lattice be $D$. Then the following bounds on entropy growth hold:
    \begin{enumerate}
        \item If the Hamiltonian $H$ is at least power-law local with decay exponent $w > 2D + 1$, then the Von Neumann entropy growth rate $|S'_1(t)| \leq c |\partial A| \bar h \log d_0$.
        \item If the Hamiltonian $H$ is superexponentially/exponentially local with $w = D \geq 1$ and $\xi < \xi_c = (\frac{1}{c_1 \log d_0})^{1/D}$, then $|S'_{\alpha \neq 1}(t)| \leq c' |\partial A| \bar h$. 
        \item If the Hamiltonian $H$ is finite-range with range $R$, and the norm of interaction terms that act nontrivially on any site is upper-bounded by $\bar h$, then $|S'_{\alpha}(t)| \leq b_{\alpha}(R)|\partial A| \bar h \log d_0$ with
        \begin{equation}
            b_1(R) \sim R^{D+1} \,, \quad b_{\alpha}(R) \sim \frac{\alpha}{|\alpha-1|} R d_0^{2c_1 R^D} \,.
        \end{equation}
    \end{enumerate}
    For all three claims, $d_0$ is the local Hilbert space dimension, $c, c',c_1$ are constants independent of $H$ and $d_0$, and $|\partial A|$ is the area of the boundary between $A, \bar A$. 
\end{theorem}
\begin{proof}
    By Lemma~\ref{lemma:key}, we only need to bound the contributions from $h_{i}(r)$ with support intersecting both $A$ and $\bar A$. Take any such term, and let $\Lambda = \mathrm{supp} \left[h_i(r)\right]$. To make use of the existing bounds on entanglement capacity, we introduce the notation $\Lambda_A + \Lambda_{\bar A} = \Lambda$ and $R_A + R_{\bar A}= \Xi \setminus \Lambda$ so that $R_A + \Lambda_A = A$ and $R_{\bar A} + \Lambda_{\bar A} = \bar A$ (see Figure~\ref{fig:LatticeSets} for a visual representation of this decomposition). 
    \begin{figure}
        \centering
        \includegraphics[width = \textwidth/2]{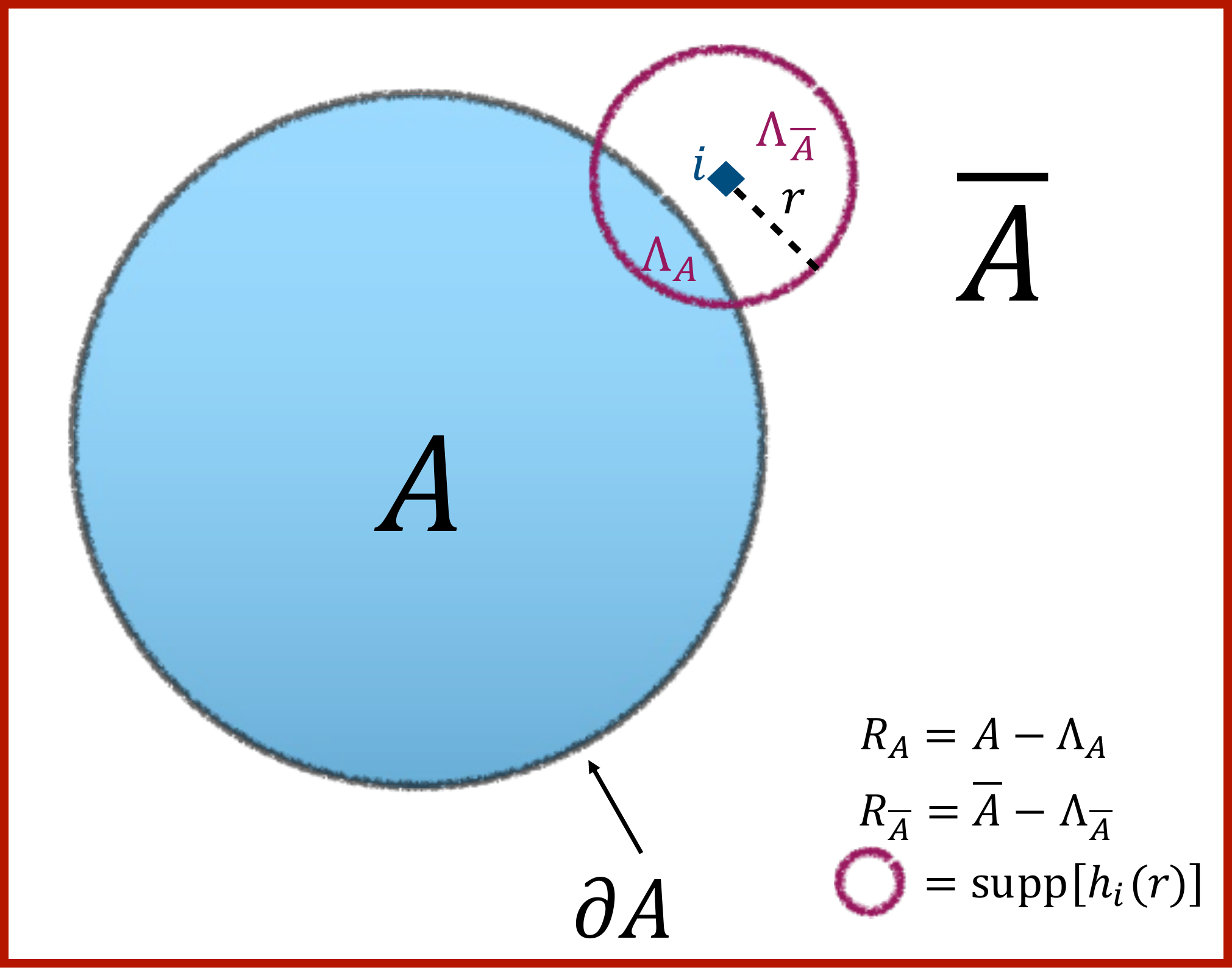}
        \caption{The red rectangle marks the boundary of the full lattice $\Xi$, which we partition into $A, \bar A$ subsystems. For each lattice site $i$, the support of the interaction term $h_i(r)$ is contained in a ball of radius $r+1$ around $i$. $\Lambda_A, \Lambda_{\bar A}$ are the intersections of the support with $A, \bar A$ respectively. Note that the interaction only contributes to the entanglement rate when $\Lambda_A, \Lambda_{\bar A}$ are both non-empty. Moreover, the number of sites contained in the support is bounded by $\mathcal{O}(r^D)$.}
        \label{fig:LatticeSets}
    \end{figure}
    First we specialize to the case of Von Neumann entropy and review the argument of Ref.~\cite{Marien_Audenaert_Acoleyen_Verstraete_2016}. For fix $i, r$ we need to bound
    \begin{equation}
        \Tr h_{i}(r) [\rho, \log \rho_A \otimes I_{\bar A}] = \Tr_{A + \Lambda_{\bar A}} h_{i}(r) [\rho_{A + \Lambda_{\bar A}}, \log \rho_A \otimes I_{\Lambda_{\bar A}}] = p^{-1} \Tr_{A + \Lambda_{\bar A}} h_{i}(r) [X, \log Y] \,,
    \end{equation}
    where $p = D_{\Lambda_{\bar A}}^{-2}$, $X = p\rho_{A + \Lambda_{\bar A}}$ and $Y = \sqrt{p} \rho_A \otimes I_{\Lambda_{\bar A}}$. By definition of density matrices, $\Tr_{A+\Lambda_{\bar A}} X = p$ and $\Tr_{A+\Lambda_{\bar A}} Y = 1$. Furthermore, since $\rho_{A+B} \leq D_B \rho_A \otimes I_B$ where $D_B$ is the dimension of the $B$ Hilbert space, we have $0 \leq X \leq Y \leq I$. This allows us to directly apply the operator inequalities in Lemma~\ref{lemma:Vershynina_ineq} and obtain\footnote{The optimal constant in the this inequality is actually 2 instead of 9, but since we are not concerned with numerical prefactors at this level of generality, we will not aim for the optimal bound here and in the rest of this section.}
    \begin{equation}
        |\Tr_{A + \Lambda_{\bar A}} h_{i}(r) [X, \log Y]| \leq ||h_{i}(r)|| \cdot ||[X,\log Y]||_1 \leq 9||h_{i}(r)|| p\log (\frac{1}{p})
    \end{equation}
    where we used the fact that $p \leq \frac{1}{2}$ for all nontrivial $\Lambda_{\bar A}$. This inequality implies that
    \begin{equation}
        |\Tr h_{i}(r) [\rho, \log \rho_A \otimes I_{\bar A}]| \leq 9||h_{i}(r)|| \log (\frac{1}{p}) = 18 ||h_{i}(r)|| \log D_{\Lambda_{\bar A}} \leq 18 c_1 ||h_{i}(r)|| \log d_0 \cdot r^D
    \end{equation}
    where $D$ is the spatial dimension of the lattice and $c_1$ is chosen so that the number of sites $i$ satisfying $d(i, \partial A) \leq r$ is upper bounded by $c_1 r^D$\footnote{This also implies that the number of sites in $\Lambda_{\bar A}$ is bounded by $c_1 r^D$ for every $\Lambda_{\bar A}$ that does not contain sites farther than $r$ apart. See Figure~\ref{fig:LatticeSets} for a visual guide.}. Summing over all contributing $h_{i}(r)$, we get 
    \begin{equation}
        \begin{aligned}
        |\frac{d S_1(t)}{dt}| &\leq 18c_1 \log d_0 \sum_{i} \sum_{r \geq d(i,\partial A)} ||h_i(r)||\, r^D \leq 18 c_1 \log d_0 \sum_r r^D \sum_{i, r \geq d(i, \partial A)} \frac{\bar h}{r^w}  \\
        &\leq 18 c_1^2 \bar h \log d_0 |\partial A| \sum_r r^{2D-w} \leq c\, \bar h\, |\partial A| \log d_0 \,.
        \end{aligned}
    \end{equation}
    The final sum over $r$ converges whenever $w > 2D+1$, which is the assumption we made. In the end, $c$ is a finite constant that depends on the lattice structure and the degree of locality, but independent of $d_0, \bar h$.

    Next, we generalize to $\alpha$-Renyi entropies. Following the same steps as before, it is easy to see that
    \begin{equation}
        \Tr h_{i}(r)[\rho, \rho_A^{\alpha-1} \otimes I_{\bar A}] = \Tr_{A+\Lambda_{\bar A}} h_{i}(r)[\rho_{A+\Lambda_{\bar A}}, \rho_A^{\alpha-1} \otimes I_{\Lambda_{\bar A}}] = D_{\Lambda_{\bar A}}^{1 + \alpha} \Tr_{A+\Lambda_{\bar A}} h_{i}(r) [X, Y^{\alpha-1}] \,,
    \end{equation}
    with the same definitions $p = D_{\Lambda_{\bar A}}^{-2}$, $X = p\rho_{A + \Lambda_{\bar A}}$ and $Y = \sqrt{p} \rho_A \otimes I_{\Lambda_{\bar A}}$. Similarly, the denominator has the structure $\Tr_A \rho_A^{\alpha} = \Tr_{A+\Lambda_{\bar A}} Y^{\alpha} \cdot D_{\Lambda_{\bar A}}^{\alpha-1}$. Hence, we obtain
    \begin{equation}
        \left|\frac{d S_{\alpha}(t)}{dt} \right| \leq |\frac{\alpha}{\alpha-1}| \left|\sum_{i} \sum_{r \geq d(i,\partial A)} p^{-1}\frac{\Tr_{A+\Lambda_{\bar A}} h_{i}(r) [X, Y^{\alpha-1}]}{\Tr Y^{\alpha}}\right| \,.
    \end{equation}
    Note that the existing operator inequalities in Lemma~\ref{lemma:Vershynina_ineq} (proven in Ref.~\cite{Vershynina_2019}) are not useful anymore. This is because the denominator in $|S'_{\alpha \neq 1}(t)|$ scales as $\Tr Y^{\alpha} \leq \Tr Y = 1$. Directly inverting this inequality gives a lower bound on $|S'_{\alpha \neq 1}(t)|$ rather than an upper bound. If we leave $\Tr Y^{\alpha}$ untouched, then the bound would be proportional to $\left[\Tr Y^{\alpha}\right]^{-1} \sim e^{(\alpha-1)S_{\alpha}(t)}$ which can become exponentially large in the size of subsystem $A$ when the $\alpha$-Renyi entropy $S_{\alpha}(t)$ has grown to be extensive. 
    
    To remove the factor of $\left[\Tr Y^{\alpha}\right]^{-1}$, we invoke Lemma~\ref{lemma:alpha>1} and obtain the stronger inequality
    \begin{equation}
        \left|\Tr_{A+\Lambda_{\bar A}} h_{i}(r) [X, Y^{\alpha-1}]\right| \leq 2 ||h_{i}(r)|| \Tr Y^{\alpha} \,.
    \end{equation}
    Part of this bound precisely cancels $\Tr Y^{\alpha}$ in the denominator. The remaining parts combine with some prefactors to give
    \begin{equation}
        \left|S'_{\alpha \neq 1}(t)\right| \leq \frac{2\alpha}{|\alpha-1|} \sum_r \sum_{i, d(i,\partial A) \leq r} D_{\Lambda_{\bar A}}^2 ||h_{i}(r)|| \leq \frac{2\alpha}{|\alpha-1|} |\partial A| \sum_r c_1 r^D d_0^{2 c_1 r^D} ||h_i(r)|| \,,
    \end{equation}
    where $c_1$ is chosen in the same way as before. Now, in order for the RHS to converge to a finite constant independent of $L$, the Hamiltonian must be at least superexponentially/exponentially local with
    \begin{equation}
        w = D \,, \quad \xi < \xi_c = \left(\frac{1}{2 c_1 \log d_0}\right)^{1/D} \,.
    \end{equation}
    When this decay condition is satisfied, we have
    \begin{equation}
        \left|S'_{\alpha \neq 1}(t)\right| \leq \frac{2\alpha}{|\alpha-1|} |\partial A| c \bar h \,,
    \end{equation}
    for some constant $c$. Redefining $c' = 2\alpha/(|\alpha-1|) c$ gives part 2 of the theorem. 
    
    Finally, for interactions with a finite range $R$, we can write
    \begin{equation}
        \begin{aligned}
            |S'_1(t)| &\leq \sum_i \sum_{d(i,\partial A) \leq R} 18 ||h_i(r)|| \log D_{\Lambda_{\bar A}}\,, \\
            |S'_{\alpha \neq 1}(t)| &\leq \frac{2\alpha}{|\alpha-1|} \sum_i \sum_{d(i, \partial A) \leq R} D_{\Lambda_{\bar A}}^2 ||h_i(r)|| \,,
        \end{aligned}
    \end{equation}
    where
    \begin{equation}
        ||h_i(r)|| \begin{cases} = 0 & r > R \\ \leq \bar h & r \leq R \end{cases} \,.
    \end{equation}
    Given the finite range of interactions, we know that $|\Lambda_{\bar A}| \leq c_1 R^D$. On the other hand, for sufficiently large $\partial A$, the number of sites $i$ satisfying $d(i,\partial A) \leq R$ is bounded by $c_2 |\partial A| R$ for some constant $c_2$ independent of the spatial dimension $D$. Given this stronger inequality, we conclude that
    \begin{equation}
        \begin{aligned}
            |S'_1(t)| &\leq 18 \left(c_2 |\partial A| R \, \bar h\right) \left(c_1 R^D \log d_0\right)\\
            |S'_{\alpha \neq 1}(t)| &\leq \frac{2\alpha}{|\alpha-1|} \left(c_2 |\partial A| R \, \bar h \right) \, \left(d_0^{2c_1 R^D}\right) \,.
        \end{aligned}
    \end{equation}
    Part 3 of the theorem is recovered when we set
    \begin{equation}
        b_1(R) = 18\,c_1\,c_2\, R^{D+1} \,, \quad b_{\alpha}(R) = \frac{2\alpha}{|\alpha-1|} c_2 \,R\, d_0^{2 c_1 R^D} \,.
    \end{equation}
\end{proof}
It is easy to generalize Theorem~\ref{thm:Bound} to lattice fermion Hamiltonians. For any bipartition of the lattice into $A + \bar A$, the Fock space factorizes as $\mathcal{H}_F = \mathcal{H}^{(A)}_F \otimes \mathcal{H}^{(\bar A)}_F$. The reduced density matrices and Renyi entropies can be defined relative to this tensor factorization and Lemma~\ref{lemma:key} continues to hold. The operator inequalities invoked in Theorem~\ref{thm:Bound} are also agnostic to the precise of choice of Hilbert space. Therefore, with the same assumptions on the locality of interactions, the analogue of Theorem~\ref{thm:Bound} holds for fermions as well. The story for bosons is more complicated because local operators on the bosonic Hilbert space do not have a well-defined operator norm. Indeed, arbitrarily fast spread of entanglement can be generated by Hamiltonians that couple nearest neighbor lattice sites with arbitrarily high boson occupation numbers. However, recent works have shown that for physically reasonable Hamiltonians and for states with low boson occupation numbers, the effective local Hilbert space can be truncated to a finite dimensional space up to small errors~\cite{TongPreskill_2021_boson_truncation}. Operator norms of bosonic operators in these reduced Hilbert spaces are bounded and the analogue of Theorem~\ref{thm:Bound} continues to hold. 
\begin{corollary}\label{cor:total_entgrowth}
    Consider the same setup as in Theorem~\ref{thm:Bound}. Then if the Hamiltonian is at least power-law local with decay exponent $w > 2D+1$, $S_{\alpha}(t) \leq c |\partial A| \bar h t$ for all $\alpha \geq 1$.
\end{corollary}
\begin{proof}
    The $\alpha$-Renyi entropies satisfy the monotonicity property: $S_{\alpha}(t) \geq S_{\beta}(t)$ whenever $\alpha \leq \beta$. By integrating part 1 of Theorem~\ref{thm:Bound}, we find that $S_{1}(t) \leq c \,|\partial A| \bar h\, t$ for all power-law local Hamiltonians with $w > 2D+1$. Hence $S_{\alpha}(t) \leq c\, |\partial A| \bar h\, t$ for all $\alpha > 1$.
\end{proof}

\subsection{Extension to systems with k-locality}

The entanglement growth bounds in Theorem~\ref{thm:Bound} admit a straightforward generalization to k-local systems which may or may not have geometric locality. More precisely, consider a quantum qudit system with $V$ sites and local Hilbert space dimension $d_0$. We say a Hamiltonian is \textbf{k-local} if $H = \sum_{i_1,\ldots i_k \in \Xi} h_{i_1,\ldots, i_k}$ where the support of $h_{i_1,\ldots, i_k}$ is contained in $\{i_1,\ldots, i_k\}$. By Lemma~\ref{lemma:key}, we have
\begin{equation}
    \left|S'_{\alpha}(t)\right| \leq \frac{\alpha}{|\alpha-1|}\left| \frac{\Tr (H_{\partial} [\rho(t), \rho_A(t)^{\alpha-1} \otimes I_{\bar A}])}{\Tr_A \rho_A(t)^{\alpha}} \right| \,,
\end{equation}
where 
\begin{equation}
    H_{\partial} = \sum_{\substack{\{i_1,\ldots,i_k\} \cap A \neq \emptyset \\ \{i_1,\ldots,i_k\} \cap \bar A \neq \emptyset}} h_{i_1,\ldots, i_k} \,.
\end{equation}
Now following the same strategy as in Theorem~\ref{thm:Bound}, we have another immediate corollary:
\begin{corollary}\label{cor:Bound_klocal}
    Consider a $k$-local Hamiltonian $H = \sum_{i_1,\ldots i_k \in \Xi} h_{i_1,\ldots, i_k}$ with $||h_{i_1,\ldots, i_k}|| \leq \bar h$ for every choice of $i_1,\ldots, i_k$. Then we have the bounds
    \begin{equation}
        |S'_{\alpha}(t)| \leq V(H_{\partial}) \bar h \begin{cases} c (k-1) \log d_0 & \alpha = 1 \\ c' d_0^{2(k-1)} & \alpha \neq 1 \end{cases} \,,
    \end{equation}
    where $V(H_{\partial})$ is the number of distinct nonvanishing terms $h_{i_1,\ldots, i_k}$ with $\{i_1,\ldots,i_k\} \cap A, \{i_1,\ldots,i_k\} \cap \bar A \neq \emptyset$. 
\end{corollary}
\begin{proof}
    The early part of the proof essentially repeats the argument in Theorem~\ref{thm:Bound}. The only difference comes from the analysis of sums over terms in $H_{\partial}$
    \begin{equation}
        |S'_{\alpha}(t)| \leq \sum_{\substack{\{i_1,\ldots,i_k\} \cap A \neq \emptyset \\ \{i_1,\ldots,i_k\} \cap \bar A \neq \emptyset}} ||h_{i_1,\ldots, i_k}|| \begin{cases} c\,\log \left( D_{\mathrm{supp}\left[h_{i_1,\ldots, i_k}\right] \cap \bar A}\right) & \alpha = 1 \\ c'\, \left(D_{\mathrm{supp}\left[h_{i_1,\ldots, i_k}\right] \cap \bar A}\right)^2 & \alpha \neq 1 \end{cases} \,,
    \end{equation}
    where $c, c'$ are constants. Now for $k$-local system, at most $k-1$ of the sites $\{i_1,\ldots, i_k\}$ can be in $\bar A$. Hence $D_{\mathrm{supp}\left[h_{i_1,\ldots, i_k}\right] \cap \bar A} \leq d_0^{k-1}$ and the corollary follows immediately.
\end{proof} 
In many interesting physical systems, the interactions between local degrees of freedom simultaneously satisfy geometric locality and k-locality. One prominent example is the Hamiltonian for electrons in a solid, where the dominant Coulomb repulsion is geometrically power-law local and $k$-local with $k = 2$. In this scenario, we can combine Corollary~\ref{cor:Bound_klocal} and Theorem~\ref{thm:Bound} to obtain a stronger set of bounds. To get there, we first define the most general form of a $k$-local Hamiltonian with geometric locality:
\begin{definition}
    For any set of sites $\{i_1,\ldots, i_k\}$ on the lattice $\Xi$ equipped with a metric $d$, define the diameter of the set by
    \begin{equation}
        \mathrm{diam} (\{i_1,\ldots, i_k\}) = \max_{a,\,b \in \{1,\,\ldots\,, k\}} d(i_a, i_b) \,.
    \end{equation}
    We say that a $k$-local Hamiltonian of the form 
    \begin{equation}
        H = \sum_{i_1,\ldots, i_k} h_{i_1,\ldots, i_k}
    \end{equation}
    is power-law local with decay exponent $w$ if 
    \begin{equation}
        ||h_{i_1,\ldots, i_k}|| \leq \frac{\bar h}{\mathrm{diam} (\{i_1,\ldots, i_k\})^w}
    \end{equation}
    for some constant $\bar h$. 
\end{definition}
Given this definition, we can state the bounds:
\begin{corollary}\label{cor:Bound_klocal+geolocal}
    Consider the same setup as in Theorem~\ref{thm:Bound} with $\partial A$ a spatially connected boundary separating $A$ and $\bar A$. Given a $k$-local Hamiltonian which is also geometrically power-law local with decay exponent $w > Dk$, we have the following bounds
    \begin{equation}
        |S'_{\alpha}(t)| \leq \begin{cases} c (k-1) \bar h \log d_0 & \alpha = 1 \\ c' d_0^{2(k-1)} \bar h & \alpha \neq 1 \end{cases} \,.
    \end{equation} 
\end{corollary}
\begin{proof}
    Applying the operator inequalities Lemma~\ref{lemma:Vershynina_ineq} and Lemma~\ref{lemma:alpha>1} directly leads to 
    \begin{equation}
        |S'_{\alpha}(t)| \leq \sum_{\substack{\{i_1,\ldots,i_k\} \cap A \neq \emptyset \\ \{i_1,\ldots,i_k\} \cap \bar A \neq \emptyset}} ||h_{i_1,\ldots, i_k}|| \begin{cases} c\,\log \left(D_{\mathrm{supp}\left[h_{i_1,\ldots, i_k}\right] \cap \bar A}\right) & \alpha = 1 \\ c'\, \left(D_{\mathrm{supp}\left[h_{i_1,\ldots, i_k}\right] \cap \bar A}\right)^2 & \alpha \neq 1 \end{cases} \,,
    \end{equation}    
    Using $D_{\mathrm{supp}\left[h_{i_1,\ldots, i_k}\right]} \leq d_0^{k-1}$, we deduce that
    \begin{equation}
        |S'_{\alpha}(t)| \leq \sum_{\substack{\{i_1,\ldots,i_k\} \cap A \neq \emptyset \\ \{i_1,\ldots,i_k\} \cap \bar A \neq \emptyset}} ||h_{i_1,\ldots, i_k}|| \begin{cases} c\,(k-1) \log d_0 & \alpha = 1 \\ c'\, d_0^{2(k-1)} & \alpha \neq 1 \end{cases} \,.
    \end{equation}
    We now make a simple geometric observation. Suppose that $\{i_1,\ldots, i_k\}$ intersects with both $A$ and $\bar A$ and there exists some $a$ for which $d(i_a, \partial A) > r$. Then there must be some choice of $b$ such that $i_a, i_b$ are on opposite sides of $\partial A$. Hence, any path between $i_a, i_b$ must intersect $\partial A$ and we have a chain of inequalities
    \begin{equation}
        \mathrm{diam}(i_1,\ldots, i_k) \geq d(i_a, i_b) \geq d(i_a, \partial A) + d(i_b, \partial A) > r \,.
    \end{equation}
    By this observation, we can replace the sum over sites by a sum over $r = \mathrm{diam}(i_1,\ldots, i_k)$ with the restriction that $d(i_a, \partial A) \leq r$ for all $a$. For fix $r$, the number of $\{i_1,\ldots, i_k\}$ satisfying $d(i_a, \partial A) \leq r$ along with the constraint $\mathrm{diam}(i_1,\ldots, i_k) = r$ is upper bounded by $c_k r^{Dk-1}$ for some $k$-dependent constant $c_k$. Thus, 
    \begin{equation}
        \sum_{\substack{\{i_1,\ldots,i_k\} \cap A \neq \emptyset \\ \{i_1,\ldots,i_k\} \cap \bar A \neq \emptyset}} ||h_{i_1,\ldots, i_k}|| \leq \sum_{r = \mathrm{diam}(i_1,\ldots, i_k)} \sum_{d(i_a,\partial A) \leq r} \frac{\bar h}{\mathrm{diam}(i_1,\ldots, i_k)^{w}} \leq \bar h \sum_r c_k r^{Dk-w-1} \,.
    \end{equation}
    Demanding convergence of the sum over $r$, we recover the constraint $w > Dk$. 
\end{proof}
The key takeaway from this new bound is that when $k$-locality is present, power-law geometric locality is sufficient for proving an $|A|$-independent bound on $|S'_{\alpha}(t)|$. This is in sharp contrast to the case with only geometric locality (see Theorem~\ref{thm:Bound}). 

\ignore{
\subsection{Extension to bosons and fermions}

So far, we have formulated our bounds in the context of qudit systems, where Hilbert spaces associated with subregions factorize. However, in many interesting physical systems, such a factorization may not exist. In this section, we extend our bounds to bosonic/fermionic lattice models. \ZS{comment more on why this is enough}

Let us consider some lattice with $N$ sites and no geometric structure a priori. On every site $i$ there are creation and annihilation operators $c^{\dagger}_i, c_i$ satisfying the standard fermion algebra (the generalization to fermions with additional internal indices is straightforward and will not be treated here)
\begin{equation}
    \{c_i, c_j\} = \{c^{\dagger}_i, c^{\dagger}_j\} = 0 \quad \{c_i, c^{\dagger}_j\} = \delta_{ij}  \,.
\end{equation}
A general Hamiltonian can be written as a double sum over lattice sites $i$ and subsets $\Lambda$ containing $i$
\begin{equation}
    H = \sum_{i, \Lambda} h_{i,\Lambda} \,,
\end{equation}
where each $h_{i,\Lambda}$ is a product of an even number of creation and annihilation operators supported in $\Lambda$. 

The next step is to define and study the subsystem Renyi entropies. Suppose we decompose the lattice into $A, \bar A$. Then the Fock space of the fermions factorizes as $\mathcal{H}_F = \mathcal{H}^{(A)}_F \otimes \mathcal{H}^{(\bar A)}_F$. Therefore, we can define the Renyi entanglement entropies as
\begin{equation}
    S^{(A)}_{\alpha}(t) =\frac{1}{1-\alpha} \log \Tr_A [\Tr_{\bar A} \rho(t)]^{\alpha} 
\end{equation}
For fixed particle number $N$, the bosonic/fermionic Hilbert space can be constructed as 
\begin{equation}
    \mathcal{H}_{\pm, N} = S_{\pm} \mathcal{H}^{\otimes N} \,,
\end{equation}
where $S_{\pm}$ correspond to the symmetrization/anti-symmetrization operator. }

\section{Illustrating the bounds in concrete models}\label{sec:new_bound_example}

In this section, we compare the bounds in Theorem~\ref{thm:Bound} and Corollaries~\ref{cor:total_entgrowth}, \ref{cor:Bound_klocal}, \ref{cor:Bound_klocal+geolocal} with the dynamics of $|S_{\alpha}(t)|$ in several concrete models with varying degrees of locality. More precisely, we want to understand how the asymptotic scaling of $|S'_{\alpha}(t)|$ with $d_0$, the interaction range $R$, and the system size $|A|$ compares with the scaling of $|S'_1(t)|$. We provide both analytic and numerical evidence that all three models obey our bounds, although none of them comes close to saturating the bounds. Possible routes to tighter bounds/saturation examples will be discussed along the way. 

\subsection{SYK Model}

We first consider the SYK model, which is a k-local but spatially non-local system described by the Hamiltonian
\begin{equation}
    H_{\rm SYK} = \sum_{i,j,k,l = 1}^N J_{ijkl} \chi_i \chi_j \chi_k \chi_l \,.
\end{equation}
Here, $\chi_i$ are $N$ flavors of Majorana fermions and $J_{ijkl}$ are a set of Gaussian i.i.d. random variables with mean zero and variance $3!J^2/N^3$. Since each interaction term in the Hamiltonian only involves four fermion operators, this model is k-local with $k = 4$. To study entanglement growth, we divide the $N$ Majorana fermions into two groups with $M, N-M$ fermions respectively. For analytic calculations, it is convenient to choose Kourkoulou-Maldacena (KM) pure states. To define this class of states, we first build $N/2$ complex fermions out of $N$ Majorana fermions via the mapping $c_j = \frac{\chi_{2j-1} + i \chi_{2j}}{2}$ and then define a set of basis states that simultaneously diagonalize the commuting number operators $\{n_j = c^{\dagger}_j c_j\}$
\begin{equation}
    (2n_j - 1) \ket{\{s\}} = s_j \ket{\{s\}}\,, \quad s_j = \pm 1 \,.
\end{equation}
For every inverse temperature $\beta$ and every string $\{s\}$, the corresponding KM state is then
\begin{equation}
    \ket{\mathrm{KM}(\{s\}, \beta)} = e^{-\beta H/2} \ket{\{s\}} \,.
\end{equation}
Starting from such an initial state, we expect the $\alpha$-Renyi entropies to grow rapidly and eventually saturate to the thermal entropy associated with the energy density of $\ket{\mathrm{KM}(\{s\}, \beta)}$. The detailed dynamics of the $\alpha$-Renyi entropies have been worked out both analytically and numerically in Ref.~\cite{zhang2020_SYKentanglement}, with the general conclusion that
\begin{equation}\label{eq:example_SYK}
    S_{\alpha}(t) \approx c N J t \quad \forall \alpha \,,\quad \forall t \ll t_{\rm saturation} \,.
\end{equation}
From \eqnref{eq:example_SYK}, we see that the entanglement growth rate is linear in both $N$ and $J$. We now ask how this linear scaling compares with the general bounds. Consider for simplicity a bipartition where $M = N/2$. Let $H_{\partial}$ be the part of the Hamiltonian that acts on both subsystems. Then $H - H_{\partial}$ consists terms in the Hamiltonian that have $i,j,k,l < N/2$ or $i,j,k,l > N/2$. This means $S(H_{\partial}) \approx N^4 - (\frac{N}{2})^4 \cdot 2 = \mathcal{O}(N^4)$. Since $\ev{J_{ijkl}^2} \approx \frac{3!J^2}{N^3}$, we conclude that the average operator norm of each term in $H_{\partial}$ is $\mathcal{O}(N^{-3/2})$. Using this estimate in Corollary~\ref{cor:Bound_klocal}, we find
\begin{equation}
    \left|S'_{\alpha}(t)\right| \leq \mathcal{O}(N^{5/2}) \cdot J \quad \forall \alpha
\end{equation}
Therefore, up to an $\mathcal{O}(1)$ multiplicative prefactor, the bound is tight in terms of $J$ scaling but loose in terms of $N$ scaling. The looseness of the bound makes sense intuitively. The $\mathcal{O}(N^{5/2})$ bound assumes complete constructive interference between all terms that contribute to $S'_{\alpha}(t)$. In the SYK model, this assumption is not valid and there are significant destructive interferences that reduce the naive expectation $\mathcal{O}(N^{5/2})$ down to $\mathcal{O}(N)$. It would be interesting to understand if this destructive interference is generic for strong quenched randomness in the couplings. If so, it is likely that our bounds can be improved with additional assumptions on the statistics of random couplings in the $k$-local interactions. 

\subsection{Mixed field Ising model}

In the previous subsection, we considered entanglement growth in the SYK model, which preserves k-locality but not spatial locality. Now, we shift gears and consider a paradigmatic spatially local chaotic system oftened referred to as the mixed field Ising model (MFIM). With open boundary conditions, the 1D MFIM Hamiltonian takes the form
\begin{equation}
    H_{MFIM} = g \sum_{i=1}^L X_i + h \sum_{i=1}^{L} Z_i + J \sum_{i=1}^{L-1} Z_i Z_{i+1} \,,
\end{equation}
where $X_i,Y_i,Z_i$ are onsite Pauli operators. If we take $\partial A$ to be the spatial cut between sites $i$ and $i+1$, then the only term in the Hamiltonian that crosses the cut is $J Z_i Z_{i+1}$. This means that the Hamiltonian is strictly finite-range with $R = 2$ and $\bar h = J$. Rerunning the proof of Theorem~\ref{thm:Bound} with this single interaction term and using the optimal constants, we obtain finite bounds
\begin{equation}\label{eq:finite_bound_MFIM}
    |S'_1(t)| \leq 2 J \log 2 \,, \quad |S'_{\alpha}(t)| \leq \frac{8\alpha}{|\alpha-1|} J \,,
\end{equation}
The important qualitative feature is that both bounds scale linearly with the coupling strength $J$. To compare these bounds with numerics, we exactly diagonalized a MFIM with $L = 16,\, g = -1.05,\, h = 0.5$ and varying values of $J$. For each $J$, we sampled 100 random product states and computed the half-chain Von Neumann entropy $S_{1}(t)$ as a function of time. The ensemble-averaged entropy growth curves are shown in Fig.~\ref{fig:MFIM_entgrowth}.
\begin{figure*}
    \centering
    \includegraphics[width = \textwidth]{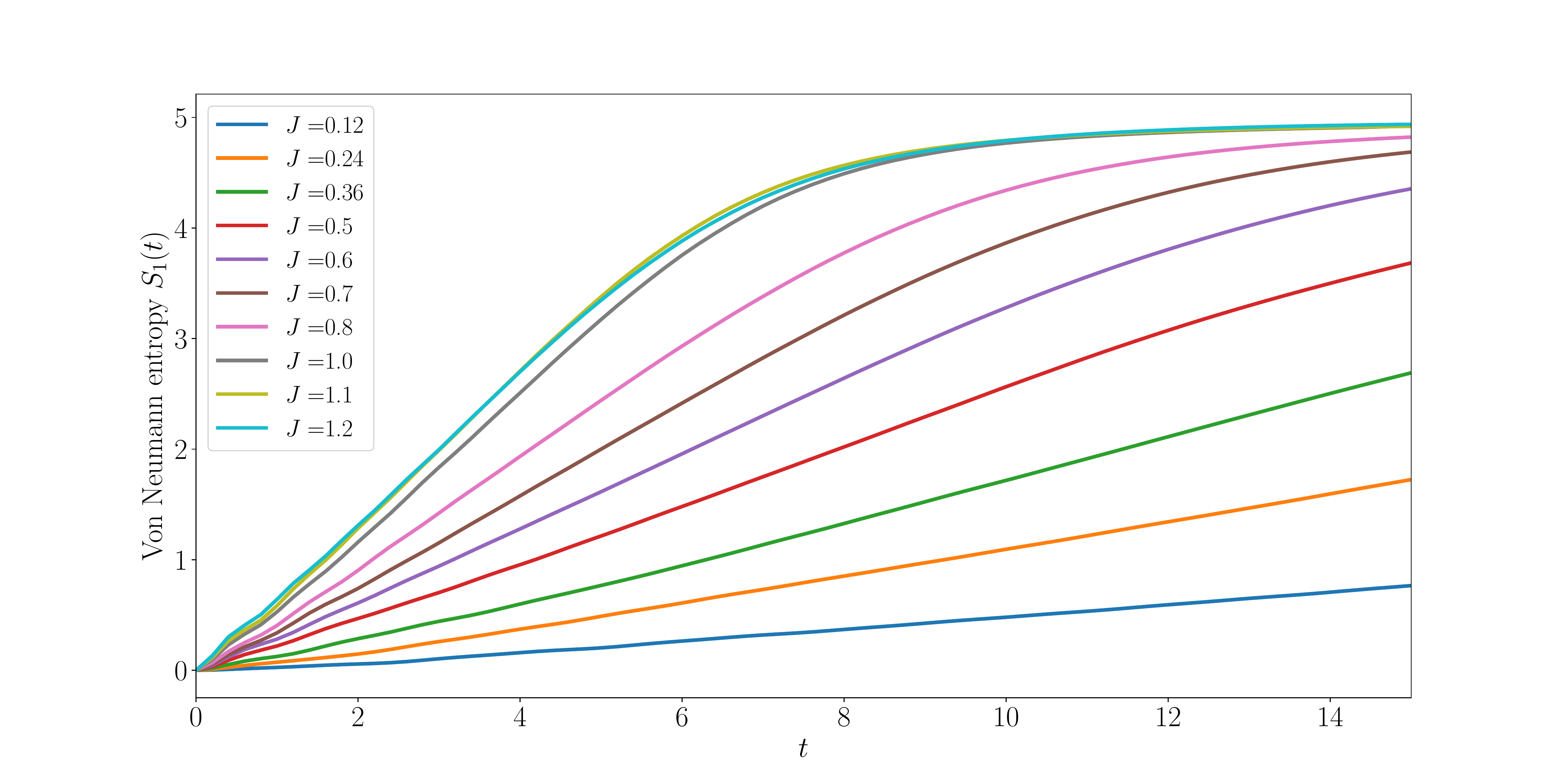}
    \caption{Time-dependence of the Von Neumann entanglement entropy for different values of $J$. The dynamics are generated by the MFIM with $L = 16$ and varying values of $J$. The entanglement entropy is calculated with respect to the middle cut. For each $J$, the plotted entropy growth curve is an average over 100 random product states. Note that after an initial transient, the entropy grows linearly until saturation.}
    \label{fig:MFIM_entgrowth}
\end{figure*}
After an early-time non-universal transient, the Von Neumann entropy grows linearly until saturation. Within the linear growth regime, we can extract the slope $s = |S'_{1}(t)|$ and the entanglement velocity $v_E = |S'_1(t)|/(\bar h \log d_0) = s(J)/(J \log d_0)$ as a function of the coupling strength $J$. The comparison plots are shown in Fig.~\ref{fig:MFIM_bound_comparison}. 

A number of interesting features can be identified. At small $J$, the empirical growth rate scales linearly with $J$, but with a slope that is smaller than the bound. Some curvature develops in the range $J \in [0.5,0.9]$. At even larger $J$, $s(J)$ appears to plateau and even begin to trend downward for $J > 1.1$. This downward trend is surprising and we do not have a physical explanation for it. In future works, it would be interesting to understand whether there exists models for which the linear scaling $s(J) \sim J$ extends to arbitrarily large values of $J$ in the thermodynamic limit (or more ambitiously, for which $s(J)$ asymptotically saturates the bound).
\begin{figure}[ht]
    \centering
    \includegraphics[width = \textwidth/2-2pt]{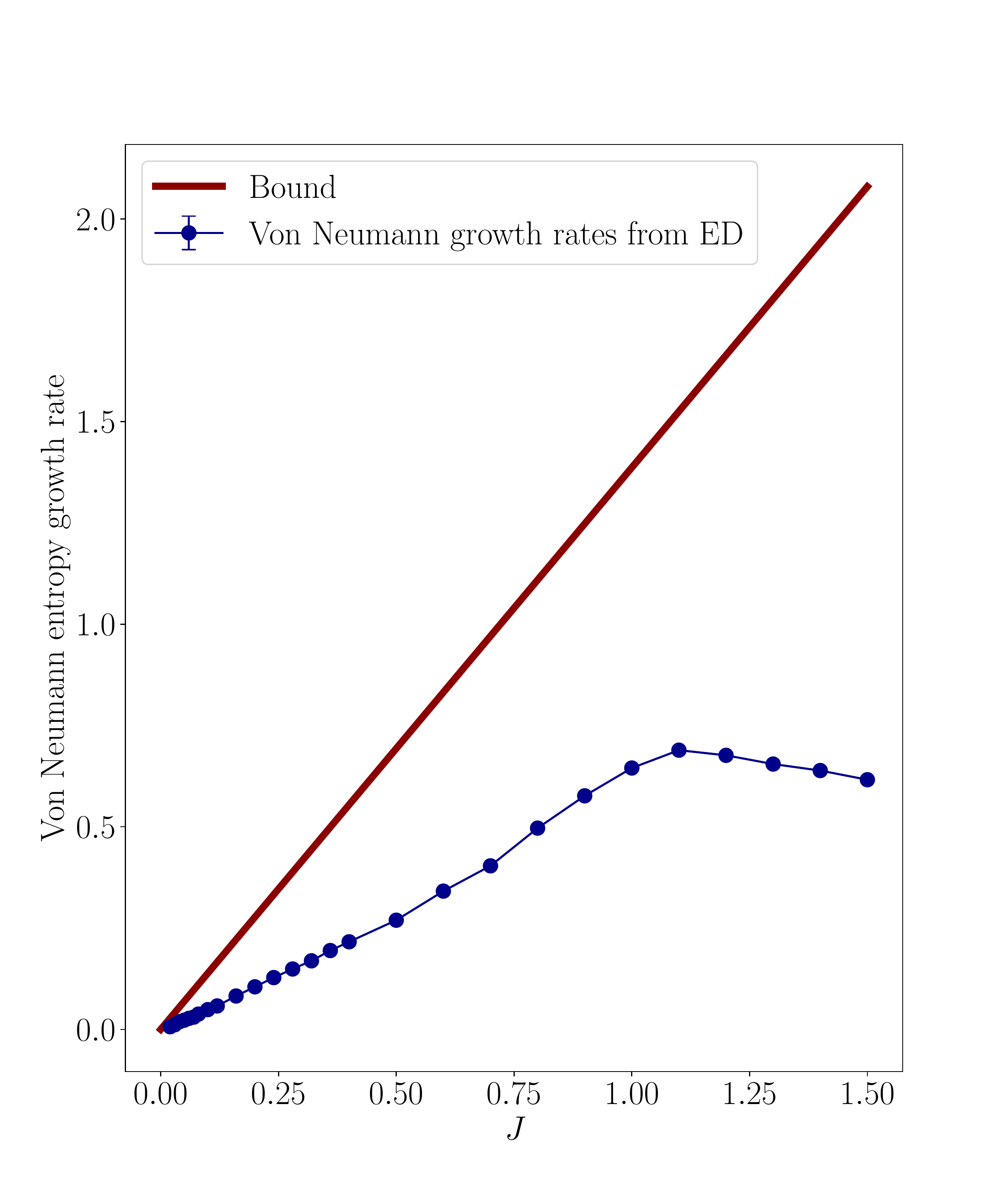}
    \includegraphics[width = \textwidth/2-2pt]{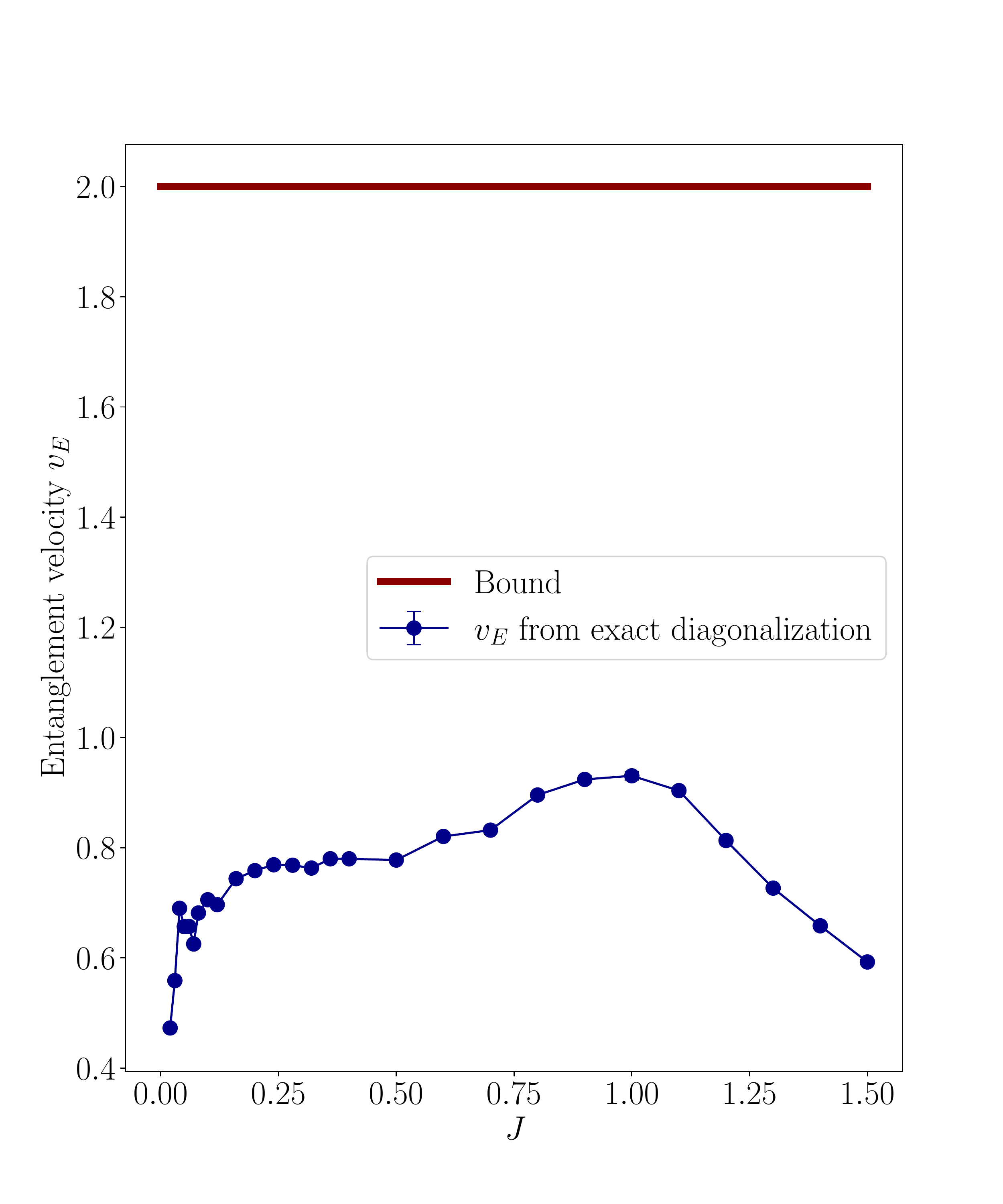}
    \caption{J-dependence of the Von Neumann entropy rate $|S'_{1}(t)|$ and the entanglement velocity $v_E$. Numerically, we take the curves in Fig.~\ref{fig:MFIM_entgrowth}, and fit the intermediate time regime to a line (the early time oscillatory regime and the late time saturation regime are nonlinear and hence excluded from the fit). The slope of each curve is the growth rate $|S'_{1}(t)|$ and the entanglement velocity $v_E = |S'_1(t)|/(J \log 2)$. The error bars for the growth rates are on the order of $10^{-3} \sim 10^{-4}$ and hence invisible on the scale of the plot.}
    \label{fig:MFIM_bound_comparison}
\end{figure}
The Renyi entropies in this model for integer $\alpha > 1$ also follow an approximate linear growth in the accessible range of $L$, though the cumulative growth at any $t$ is smaller than $S_1(t) - S_1(0)$ due to Corollary~\ref{cor:total_entgrowth}. Unlike the case of random GUE Hamiltonians considered in Section~\ref{sec:randomGUE_example}, we do not find any singularities in $|S'_{\alpha}(t)|$ as a function of $t$. This is consistent with \eqnref{eq:finite_bound_MFIM}.

\subsection{SYK Chain}

As a final example, we consider a 1D chain of doubled SYK models where fermions at different lattice sites are coupled by nearest-neighbor random interactions. Concretely, the Hamiltonian takes the form
\begin{equation}
    H_{\rm SYK \,chain} = H_{\uparrow} + H_{\downarrow}  \,,
\end{equation}
\begin{equation}
    H_{\alpha = \uparrow/\downarrow} = \sum_{ijkl=1}^N \sum_{x=1}^L J^{\alpha}_{0,ijkl, x} \chi^{\alpha}_{i,x} \chi^{\alpha}_{j,x} \chi^{\alpha}_{k,x} \chi^{\alpha}_{l,x} + \sum_{ijkl=1}^N \sum_{x=1}^L J^{\alpha}_{1,ijkl, x} \chi^{\alpha}_{i,x} \chi^{\alpha}_{j,x} \chi^{\alpha}_{k,x+1} \chi^{\alpha}_{l,x+1}
\end{equation}
where $J^{\alpha}_{0,ijkl,x}, J^{\alpha}_{1,ijkl,x}$ are i.i.d. Gaussian random variables with 
\begin{equation}
    \ev{J^{\alpha}_{0,ijkl,x}} = \ev{J^{\alpha}_{1,ijkl,x}} = 0 \,, \quad \ev{(J^{\alpha}_{0,ijkl,x})^2} = \frac{J_0^2}{N^3} \,, \quad \ev{(J^{\alpha}_{1,ijkl,x})^2} = \frac{J_1^2}{N^3} \,.
\end{equation}
Here, $x$ is a label for the spatial sites and $N$ is the total number of fermion flavors within each site. Importantly, we impose spatial locality for interactions between different lattice sites, but only k-locality for interactions within a lattice site. This setup allows us to test the scaling of entanglement growth rates with the onsite Hilbert space dimension $d_0 \sim 2^N$, the interaction strength $\bar h$ and the lattice size $L$. 

Remarkably, the growth of $\alpha$-Renyi entropies in the above 1D SYK chain has been solved analytically in the large $N,L$ limit in Ref.~\cite{Gu_Lucas_Qi_2017} for a special class of initial product states that we now define. Let $\ket{EPR}_x$ be a maximally entangled state between the fermions in the $\uparrow$ and $\downarrow$ family on a fixed spatial site $x$. Then $\ket{I} = \otimes_{x=1}^L \ket{EPR}_x$ is state with no spatial entanglement across any spatial cut and hence a good candidate for studying entanglement growth at infinite temperature. To vary the effective temperature, we can also generalize this state to a thermofield double state $\ket{TFD} = e^{-\frac{\beta}{4} (H_{\uparrow} + H_{\downarrow})}$, which is also spatially unentangled. For this modified setup, and within the perturbative limit $J_1 \ll \frac{J_0}{\sqrt{\beta J_0}}$, the authors of Ref.~\cite{Gu_Lucas_Qi_2017} then computed the $\alpha$-fold replicated partition function $\Tr \rho_A(t)^{\alpha}$ for a suitable initial density matrix and then performed an analytic continuation to real values of $\alpha$. Keeping only replica-diagonal contributions to $\Tr \rho_A(t)^{\alpha}$, the growth of $S_{A,\alpha}(t)$ for general $\alpha > 1$ is found to be
\begin{equation}\label{eq:example_SYKchain}
    S_{A,\alpha}(t) \approx \frac{\alpha}{\alpha-1} \frac{2\pi}{\beta} \frac{J_1^2 N}{8 \pi J_0^2} t \,.
\end{equation}
The singularity of this functional form as we approach the Von Neumann entropy limit $\alpha \rightarrow 1$ is clearly unphysical. To restore a physical answer for $\alpha = 1$ and for $\alpha < 1$, one would need to include replica off-diagonal contributions that cancel divergent parts of the diagonal solution. Since a full solution of $S_{A,1}(t)$ is not available, we will simply compare our bounds with the analytic results for $\alpha \neq 1$. 

One surprise in this model is that the exact entanglement velocity scales as $J_1^2$ at small $J_1$ which is much smaller than the bound proportional to $J_1$. 
One plausible reason for this scaling discrepancy is that the analytic result \eqnref{eq:example_SYKchain} holds only for initial states with $\beta \neq 0$, while the rigorous bounds hold for all $\beta$. As $\beta \rightarrow 0$, the analytic result breaks down and we conjecture that the scaling of entanglement growth rate would begin to track the linear bound. It would be interesting to extend the analytic results to these $\beta = 0$ states and test this conjecture. 

Another interesting feature is the linear $N$-scaling of $S'_{A,\alpha}(t)$. Recall that the SYK Hamiltonian on a single spatial site has an $\mathcal{O}(N)$ operator norm due to the choice of coupling distribution. The same spatial site has a Hilbert space dimension $d_0 = 2^{N}$. With these estimates, Theorem~\ref{thm:Bound} implies $|S'_{\alpha>1}(t)| \lesssim c N 2^{2N}$, while Corollary~\ref{cor:total_entgrowth} implies $|S_{\alpha}(t)| \lesssim c' N t$. Up to a multiplicative constant, the analytic cumulative entanglement growth $S_{\alpha}(t) - S_{\alpha}(0)$ saturates the bound in Corollary~\ref{cor:total_entgrowth}. However, the incremental entanglement rate $\frac{\alpha}{\alpha-1} \frac{2\pi}{\beta} \frac{J_1^2 N}{8\pi J_0^2}$ is exponentially smaller than the bound $c N 2^{2N}$. This separation of scale leaves us with two very different possibilities: (1) the bound for $\max_t |S_{\alpha > 1}'(t)|$ in Theorem~\ref{thm:Bound} is far from optimal and (2) the bound is tight, but the SYK chain Hamiltonian is too simple to saturate the bound. One reason why (2) may be plausible is that the onsite interactions of the SYK chain are generated by k-local sparse matrices that may not be able to produce entanglement at a maximal rate. To resolve this uncertainty, it would be useful to replace each SYK dot with a featureless dense GUE random matrix. For example, one toy model that one could consider is the 1D GUE-chain
\begin{equation}
    H_{\rm GUE-chain} = \sum_{x=1}^L M_x + \sum_{x=1}^{L-1} M_x M_{x+1} + h.c. \,,
\end{equation}
where each $M_x$ is drawn from an independent rank-$N$ GUE distribution. Given the solvability of entropy dynamics generated by a single random GUE Hamiltonian, there is some hope that certain properties of the GUE-chain can also be determined analytically. It would be interesting to see if $|S'_{\alpha \neq 1}(t)|$ in this model can get closer to the bound in Theorem~\ref{thm:Bound}. 

\section{Discussion}\label{sec:discussion}

The central motivation of this work is to explore similarities and differences between the dynamics of $\alpha$-Renyi entropies $S_{\alpha \neq 1}(t)$ and Von Neumann entropy $S_1(t)$, two of the most popular measures of bipartite entanglement in many-body quantum systems. 
The approach we take is to prove universal (state-independent) bounds (Lemma~\ref{lemma:alpha>1}, Theorem~\ref{thm:Bound}, Corollaries~\ref{cor:total_entgrowth}, \ref{cor:Bound_klocal}, \ref{cor:Bound_klocal+geolocal}) on the time derivatives of these entropy measures and compare how each of the bounds scales with the subsystem size, the interaction strength, the degree of locality in the Hamiltonian etc. The key observation we make is that this similarity between $S_1(t)$ and $S_{\alpha}(t)$ breaks down in general, although it appears to hold for many solvable models. Concretely, no matter which notion of locality we impose on the Hamiltonian, there is always a gap between the optimal bounds that we could prove for $\alpha > 1$ and for $\alpha = 1$\footnote{We emphasize that the tightness of these bounds is largely unknown, except in the case of completely non-local dynamics.}. For the bipartite entanglement between complementary subsystems $A$ and $\bar A$, these gaps can be summarized as follows:
\begin{enumerate}
    \item With no locality assumptions, the bound on $|S'_{1}(t)|$ scales linearly with the subsystem volume $|A|$, while the bound on $|S'_{\alpha\neq1}(t)$ scales exponentially with $|A|$. Moreover, this \textit{exponential separation can be saturated}. 
    \item With geometric quasi-locality (infinite-range interactions that decay with distance), a finite bound on $|S'_{\alpha}(t)|$ requires power-law local interactions for $\alpha = 1$ but superexponentially local interactions for $\alpha \neq 1$. 
    \item With strict geometric locality (finite-range interactions), the bound on $|S'_{\alpha}(t)|$ scales as a power law of the interaction-range $R$ for $\alpha = 1$ and as an exponential of $R^D$ for $\alpha \neq 1$. 
    \item With $k$-locality (and with/without geometric locality), the bound on $|S'_{\alpha}(t)|$ scales as a power law of $k$ for $\alpha = 1$ and as an exponential of $k$ for $\alpha \neq 1$. 
\end{enumerate}
If these bounds are tight, then they imply a qualitative distinction between the dynamics of $S_{\alpha}(t)$ and $S_1(t)$ that is in tension with the ``entanglement membrane" picture~\cite{Nahum_Ruhman_Vijay_Haah_2017,jonay2018_membrane,Zhou_Nahum_2019,ZhouNahum2020_membrane_chaotic}. For example, in the membrane picture, one associates a natural velocity scale $v_E$ to a general entanglement measure $E(t)$ via the equation
\begin{equation}
    \left|\frac{dE(t)}{dt} \right| = v_E |\partial A| s_{th}(\rho_0) \,,
\end{equation}
where $|\partial A|$ is the size of the boundary between $A, \bar A$ and $s_{th}(\rho_0)$ is the equilibrium entropy density associated with the initial state $\rho_0$. For generic non-integrable Hamiltonians with no global symmetry, this expectation is justified by approximating the local dynamics generated by the Hamiltonian with local random unitary circuits. For these circuits, $S_{\alpha}(t)$ can be mapped to the free energy of fluctuating domain walls on a spacetime lattice with length $t$ in the temporal direction. As $t$ increases, the free energy of the domain walls (and hence the entropies) increases linearly with $t$. Since different choices of $\alpha$ only modify the domain wall structures and their effective interactions, the basic linear-in-$t$ scaling should be universal, while the entanglement velocities $v_E(\alpha)$ could depend on $\alpha$~\cite{Zhou_Nahum_2019}. However, if our bounds are tight, we could design a geometrically quasi-local Hamiltonian for which $|S'_1(t)|$ is bounded for all initial states but $\max_t |S'_{\alpha \neq 1}(t)|$ can be arbitrarily large in the $|A| \rightarrow \infty$ limit. Then the notion of entanglement velocity makes sense for $S_1(t)$ but not for $S_{\alpha}(t)$ (in fact, for Von Neumann entropy, Theorem~\ref{thm:Bound} immediately implies an upper bound on the entanglement velocity $v_E \leq c \bar h$). Such a conclusion would challenge the membrane picture for $\alpha$-Renyi entropies in general\footnote{In systems with global energy/charge conservation, $S_{\alpha}(t)$ actually scales as $\sqrt{t}$ at large times due to energy/charge diffusion~\cite{Rakovszky_Pollmann_vonKeyserlingk_2019,Huang_2020} and $v_E(\alpha > 1)$ cannot be defined. This phenomenon already requires a modification of the membrane picture. However, since $|S'_{\alpha}(t)| \ll |S'_1(t)|$ at large $t$, examples in this class are even further from saturating our bounds.}.

The preceding discussion makes it clear that the tightness of our bounds is an important open question. In the case of non-local Hamiltonians, we were able to show that our bounds are almost saturated by random GUE dynamics. However, for geometrically local and k-local Hamiltonians, we have not found solvable models that exhibit the kind of exponential separation anticipated by the bounds. Can the bound for $|S_{\alpha}'(t)|$ be saturated? Or is the bound just an artifact of our proof techniques? One interesting family of models where these saturation questions can be explored are GUE-chains where independently random GUE Hamiltonians are placed on each site of a lattice and coupled by quasi-local interactions (ranging from power-law local to strictly finite-range). The solvability of entropy dynamics for a single GUE site may provide some analytic insights for tackling this more general family of models. On the flip side, one can ask if the bounds can be strengthened with additional assumptions on the structure of local terms or on the ensemble of initial states. In the proof of Corollary~\ref{cor:Bound_klocal+geolocal}, we have seen how a combination of $k$-locality and geometric locality leads to bounds on $|S'_{\alpha}(t)|$ that are stronger than Theorem~\ref{thm:Bound} for all values of $\alpha$. Perhaps this line of thinking can be pushed further with more sophisticated operator inequalities.

\section*{Acknowledgements}
    We thank Shankar Balasubramanian, Sarang Gopalakrishnan, Hong Liu, Daniel K. Mark, Daniel Ranard, Shreya Vardhan, and Romain Vasseur for helpful discussions and comments on the draft. We also thank Hong Liu and Shreya Vardhan for collaboration on related topics. ZDS is supported by the Jerome I. Friedman Fellowship Fund, as well as in part by the Department of Energy under grant DE-SC0008739.
    
\appendix

\nocite{apsrev41Control}
\bibliographystyle{apsrev4-1}
\bibliography{entbounds}

\end{document}